\documentclass{llncs}

\usepackage{amsmath,amsfonts,amssymb}
\usepackage{latexsym}
\usepackage{color}
\usepackage[normalem]{ulem}
\usepackage{enumerate}
\usepackage{eucal,amscd, graphics, graphicx}
\usepackage[ruled,noline,noend,linesnumbered]{algorithm2e}
\usepackage{pseudocode}
\usepackage[subrefformat=parens,labelformat=parens]{subfig}
\usepackage{lscape}
\usepackage{mathtools}
\usepackage{hyperref}
\usepackage{rotating}
\usepackage{cite}
\usepackage{multirow}
\usepackage{geometry} 
\usepackage[verbose]{wrapfig}

\newcommand{\C}{\mathbb{C}}	            

\newcommand{\bunderline}[1]{\underline{#1\mkern-4mu}\mkern4mu}
\newcommand{\boverline}{\bar}

\newcommand{\ProblemFull}{\xspace{Perfect Phylogeny Mixture Deconvolution Problem}}
\newcommand{\ProblemCladisticFull}{\xspace{Cladistic Perfect Phylogeny Mixture Deconvolution Problem}}

\newcommand{\Problem}{\xspace{PPMDP}}

\newcommand{\ProblemCladistic}{\xspace{Cladistic-PPMDP}}

\spnewtheorem{observation}{Observation}{\bfseries}{\itshape}
\spnewtheorem{invariant}{Invariant}{\bfseries}{\itshape}
\spnewtheorem{inv}{(Main Text) Invariant}{\bfseries}{\itshape}
\spnewtheorem{lemU}{(Main Text) Lemma}{\bfseries}{\itshape}
\spnewtheorem{def1}{(Main Text) Definition}{\bfseries}{\itshape}
\spnewtheorem{thm1}{(Main Text) Theorem}{\bfseries}{\itshape}
\spnewtheorem{thm2}{(Main Text) Theorem}{\bfseries}{\itshape}
\spnewtheorem{thm3}{(Main Text) Theorem}{\bfseries}{\itshape}
\spnewtheorem{lem2}{(Main Text) Lemma}{\bfseries}{\itshape}
\spnewtheorem{lem4}{(Main Text) Lemma}{\bfseries}{\itshape}
\spnewtheorem{prop1}{(Main Text) Proposition}{\bfseries}{\itshape}
\spnewtheorem{prob1}{(Main Text) Problem}{\itshape}{}

\author{Mohammed El-Kebir\inst{1} \and Gryte Satas\inst{1} \and Layla Oesper\inst{1,2} \and Benjamin J.~Raphael\inst{1}}
\title{Multi-State Perfect Phylogeny Mixture Deconvolution and Applications to Cancer Sequencing}
\titlerunning{Perfect phylogeny factorization problem}

\institute{$^1$Department of Computer Science and Center for Computational Molecular Biology, Brown University, Providence, RI 02912. $^2$Department of Computer Science, Carleton College, Northfield, MN 55057.
}

\begin{document}

\maketitle

\begin{abstract}
The reconstruction of phylogenetic trees from mixed populations has become important in the study of cancer evolution, as sequencing is often performed on bulk tumor tissue containing mixed populations of cells.
Recent work has shown how to reconstruct a perfect phylogeny tree from samples that contain mixtures of two-state characters, where each character/locus is either mutated or not.  However, most cancers contain more complex mutations, such as copy-number aberrations, that exhibit more than two states.  We formulate the Multi-State \ProblemFull\ of reconstructing a multi-state perfect phylogeny tree given mixtures of the leaves of the tree.  We characterize the solutions of this problem as a restricted class of spanning trees in a graph constructed from the input data, and prove that the problem is NP-complete. 
We derive an algorithm to enumerate such trees in the important special case of cladisitic characters, where the ordering of the states of each character is given.  We apply our algorithm to simulated data and to two cancer datasets.  On simulated data, we find that for a small number of samples, the Multi-State \ProblemFull\ often has many solutions, but that this ambiguity declines quickly as the number of samples increases.  On real data, we recover copy-neutral loss of heterozygosity, single-copy amplification and single-copy deletion events, as well as their interactions with single-nucleotide variants.
\end{abstract}
\section{Introduction}
Many evolutionary processes are modeled using phylogenetic trees, whose 
leaves correspond to extant entities called \emph{taxa}, and whose edges describe the ancestral relationships between the taxa.
In a character-based phylogeny, taxa are represented by a collection of \emph{characters}, each of which have one of several distinct \emph{states}.
The basic problem of character-based phylogenetic tree reconstruction is the following:  given the states for $n$ characters on $m$ taxa, find a vertex-labeled tree whose leaf labels are the taxa and whose internal nodes are labeled by an ancestral state for each character such that the resulting tree maximizes an objective function (e.g. maximum parsimony or maximum likelihood) over all such labeled trees.

Here, we consider a phylogenetic tree mixture problem, where the input is not the set of states for each taxon, but rather \emph{mixtures} of these states (Figure~\ref{fig:overview}(a)).  The goal is to explain these given mixtures as a mixing of the leaves of an unknown phylogenetic tree in unknown proportions (Figure~\ref{fig:overview}(b)).  Stated another way:   if one is given $m$ mixtures of the leaves of a phylogenetic tree, can one reconstruct the tree and the mixing proportions? This problem is motivated by cancer genome sequencing.  Cells in a tumor are derived from a single founder cell and are distinguished by somatic mutations~\cite{Nowell:1976aa}.  However, due to technical and financial constraints, most cancer sequencing projects do not sequence individual cells~\cite{Navin:2014aa,Wang:2014aa,Navin:2015aa,Zhang:2015aa}, but rather bulk tumor samples containing thousands to millions of cells~\cite{Cancer-Genome-Atlas-Research-Network:2013aa,Zhang:2011aa}.  Thus, the data we observe from a sequencing experiment is the fraction of reads that indicate a mutation, which is proportional to the fraction of cells that contain the mutation.

The difficulty of the phylogenetic tree mixture problem depends on the evolutionary model.  The simplest such model assumes that characters are binary (i.e.\ have two states) and change state at most once (i.e.\ the characters evolve with no homoplasy).
In the resulting phylogenetic tree each character-state pair thus labels at most one edge.
Such a restricted phylogenetic tree is called a \emph{perfect phylogeny}.
The no-homoplasy assumption is also called the \emph{infinite sites assumption} for two-state characters. 
Deciding whether a set of taxa with two-state characters admits a perfect phylogeny is solvable in polynomial time~\cite{Estabrook:1975ew,Gusfield:1991}.

Recently, driven by the application to cancer sequencing, there has been a surge of interest in solving the phylogenetic tree mixture problem for a two-state perfect phylogeny, relying on the idea that the infinite sites assumption is a reasonable model for somatic single nucleotide variants (SNVs) \cite{Nik-Zainal:2012aa,Strino:2013,Jiao:2014aa,Hajirasouliha:2014aa,Malikic:2015aa,Popic:2015,ElKebir:2015by}.  All of this work is based on the observation that the infinite sites assumption provides strong constraints  on the possible ancestral relationships between any pair of characters, given their observed frequencies across all samples.
In particular, \cite{ElKebir:2015by,Popic:2015} showed that the phylogenetic trees that produce the observed frequencies correspond to constrained spanning trees of a certain graph, with \cite{ElKebir:2015by} proving that when there are no errors in the measured frequencies there is a 1-1 correspondence between phylogenetic trees and these constrained spanning trees.

While the two-state perfect phylogeny model might be reasonable for somatic SNVs, it is fairly restrictive for modeling the somatic mutational process in cancer.   For instance, somatic copy-number aberrations (CNAs) are ubiquitous in solid tumors~\cite{Zack:2013aa}, and these mutations generally have more than two states.  While in some cases it may be possible to exclude CNAs in tree reconstruction, there can be interesting interactions between SNVs and CNAs that confound such analyses \cite{Eirew:2015cg}. Deshwar et  al.~\cite{Deshwar:2015kw} introduced a probabilistic graphical model for SNVs and CNAs but do so by modeling CNAs as special characters of a two-state perfect phylogeny, rather than addressing the general multi-state problem.
A probalistic model introduced by Li and Li~\cite{Li:2014bp} considers SNVs and CNAs to infer the genomic composition of tumor populations without inferring a tree describing their evolutionary relationships.
In another line of work, Chowdhury et al.~\cite{Chowdhury:2015kd} have developed rich models for copy-number aberrations, but these are for single-cell data and do not address the phylogenetic tree mixture problem.  

For multi-state characters, the no homoplasy assumption is referred to as the \emph{infinite alleles assumption} in population genetics, or the  \emph{multi-state perfect phylogeny}~\cite{Fernandez:2001,gusfield2014recombinatorics}.  In this model, a character may change state more than once on the tree, but changes \emph{to the same state at most once}.
In contrast to the case of two-state characters, the multi-state perfect phylogeny problem is NP-complete~\cite{DBLP:conf/icalp/BodlaenderFW92} in general, but is fixed-parameter tractable in the number of states per character~\cite{Agarwala:1994gd,Kannan96afast}.  There is an elegant connection between multi-state perfect phylogeny and restricted triangulations of chordal graphs \cite{Buneman1974}, which was recently exploited by Gusfield and collaborators to obtain combinatorial conditions for the multi-state perfect phylogeny \cite{GusfieldJCB2010,Gysel:2011aa}.

\paragraph{\textbf{Contributions.}}

Here, we introduce a perfect phylogeny mixture problem in the case of multi-state characters that evolve without homoplasy under the infinite alleles assumption. We define this problem formally as the Multi-State \ProblemFull~(PPMDP). We derive a characterization of the solutions of the PPMDP as a restricted class of spanning trees of a labeled multi-graph that we call the \emph{multi-state ancestry graph}.  Using this characterization, we show that the PPMDP is NP-complete. We adapt the Gabow-Myers~\cite{Gabow:1978cg} algorithm to enumerate these trees, allowing for errors in the input.
We apply our algorithm to simulated data and find that there is considerable ambiguity in the solution with sequencing data from few samples, but that the number of phylogenetic trees decreases dramatically as the number of samples increases.  We apply our algorithm to two cancer datasets, a chronic lymphocytic leukemia (CLL) and a prostate tumor, and infer phylogenetic trees that contain copy-number aberrations (CNAs), single nucleotide variants (SNVs), and combinations of these.


\begin{figure}[t]
  \centering
  \includegraphics[width=.7\textwidth]{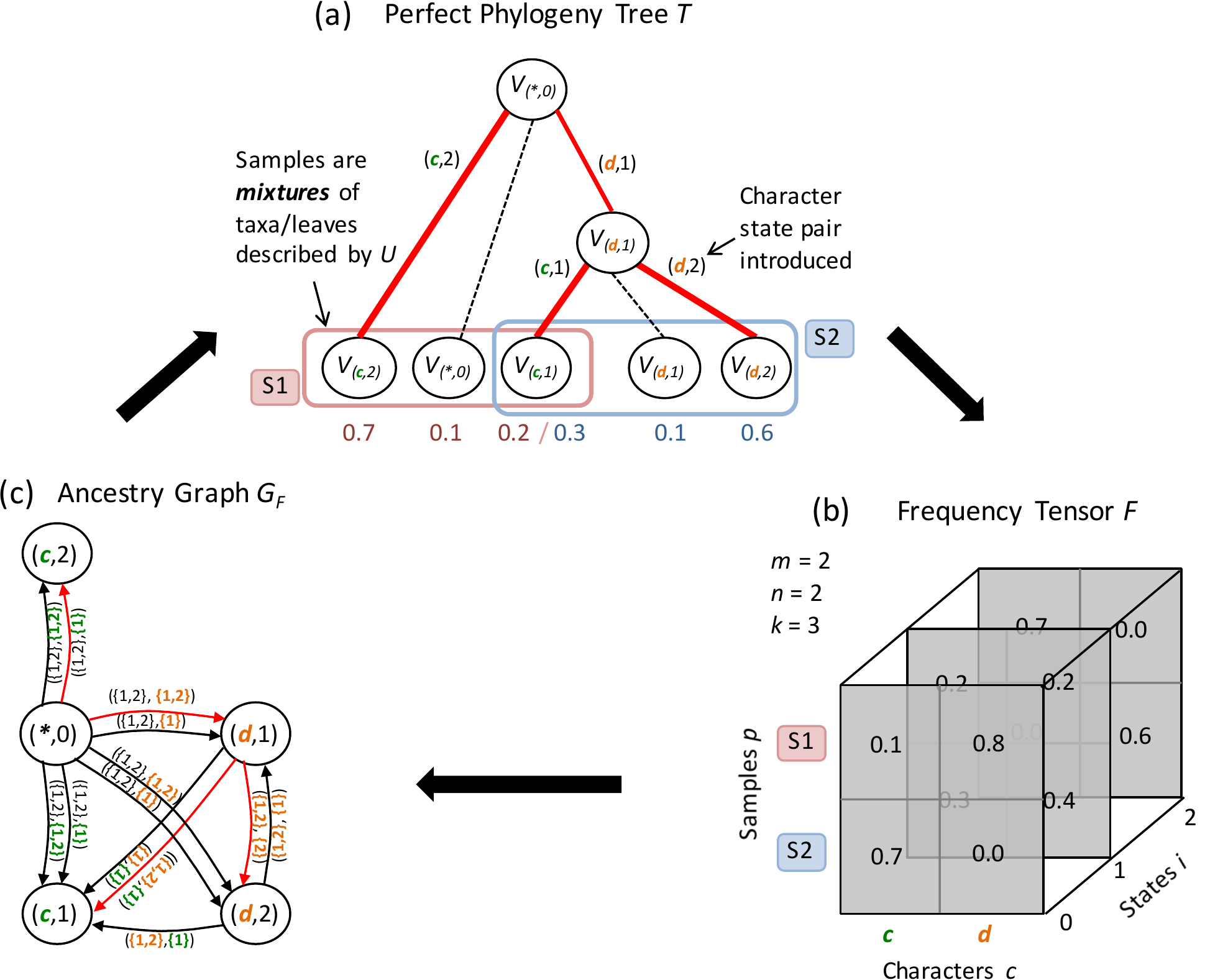}
  \caption{
    \textbf{Overview of the \ProblemFull:} 
    (a) A multi-state perfect phylogeny tree $T$ whose leaves are mixed according to proportions $U$.
    (b) This results in the frequency tensor $\mathcal{F}$ whose entries are the observed frequency of character-state pairs across samples. 
    The \Problem\ asks to infer $T$ and $U$ given $\mathcal{F}$.
    (c) The multi-state ancestry graph $G_{\mathcal{F}}$ describes potential ancestral relationships between character-state pairs.  A \textit{threading} through this graph (red edges) produces a perfect phylogeny tree~$T$. 
 }
  \label{fig:overview}
\end{figure}

\section{The \ProblemFull}
\label{sec:problem_statement}
  
Let $A \in \{0,\ldots,k-1\}^{m \times n}$ be a matrix whose rows correspond to $m$ taxa and whose columns represent $n$ characters, each of which has at most $k$ states.
A tree $T$ is a \emph{perfect phylogeny tree} for $A$ provided (1) each vertex is labeled by a \emph{state vector} in $\{0,\ldots,k-1\}^n$, which denotes the state for each character; (2) each leaf corresponds to exactly one row (taxon) of $A$ labeled by the same states for each character; (3) vertices labeled with state $i$ for character $c$ form a connected subtree $T_{(c,i)}$ of $T$~\cite{gusfield2014recombinatorics,Fernandez:2001}.
Throughout this paper, we restrict our attention to the case of rooted, or directed, perfect phylogenies, where the state vector for the root vertex is all zeros (Figure~\ref{fig:overview}(a)).  This restricted version is consistent with the assumptions for our cancer sequencing application below.  Moreover, perfect phylogeny algorithms can generally be extended to the non-zero rooted case or unrooted case with some additional work~\cite{gusfield2014recombinatorics}.

Rather than measuring $A$ or the leaves of $T$ directly, we are given $m$ samples that are \emph{mixtures} of the rows (taxa) of $A$, or equivalently mixtures of the labels of the leaves of $T$, in unknown proportions.  
Consider a sample $p$ and a character $c$ with states $i \in \{0,\ldots,k-1\}$.
The proportion of taxa in sample $p$ that have state $i$ for character $c$ is given by $f_{p,(c,i)}$.
Note that $f_{p,(c,i)} \ge 0$ and $\sum_{i} f_{p,(c,i)} = 1$.
Our input measurements are thus described by a
$k \times m \times n$ \emph{frequency tensor} $\mathcal{F} = [[f_{p, (c,i)}]]$---see Figure~\ref{fig:overview}(b).

While the frequencies $\mathcal{F}$ are obtained by mixing the leaves of $T$,  the number of leaves in $T$ is typically unknown.  
Since any internal vertex can be extended without change to a leaf and mixing proportions may be zero, we can assume without loss of generality that $\mathcal{F}$ is obtained by mixing the vertices (instead of the leaves) of an \emph{$n,k$-complete perfect phylogeny tree} $T$ with exactly $n(k-1)+1$ vertices representing the possible state vectors.
See Supplementary Methods~\ref{sec:app_preliminaries} for further details.

In an $n,k$-complete perfect phylogeny tree $T$, each non-root vertex of $T$ is the root of exactly one subtree $T_{(c,i)}$.
The root vertex of $T$ corresponds to the root of each subtree in $\{ T_{(c,0)} \mid c \in [n] \}$.
We thus label each vertex of $T$ by $v_{(c,i)}$ and the root vertex of $T$ by $v_{(*,0)}$.
Further, we use $\delta(c,i)$ to denote the children of $v_{(c,i)}$ and we write $(c,i) \prec_T (d,j)$ if $v_{(c,i)}$ is ancestral to $v_{(d,j)}$ in $T$. Note that $\prec_T$ is a reflexive relation.
For each vertex $v_{(c,i)}$ and each sample $p$, we define the mixing proportion $u_{p, (c,i)}$ to be the proportion of sample $p$ that is from the state vector at $v_{(c,i)}$.  Note that we have $u_{p, (c,i)} \ge 0$ for all $p$, $c$, and $i$, and $\sum_{(c,i)} u_{p, (c,i)} = 1$ for all $p$.
We refer to $u_{p, (c,i)}$ as the \emph{usage} of vertex $v_{(c,i)}$ in sample $p$, and define an $m,n,k$-\emph{usage matrix} $U = [u_{p, (c,i)}]$.   
It follows that $f_{p, (c,i)} = \sum_{(d,j) \in T_{(c,i)}} u_{p, (d,j)}$.  
Figure~\ref{fig:overview}(b) shows a perfect phylogeny tree $T$ and how it may be mixed according to a usage matrix $U$ to yield the observed frequency tensor $\mathcal{F}$. 
We have the following problem.
\begin{problem}[\ProblemFull~(\Problem)]
  \label{prob:1}
  Given an $k \times m \times n$ frequency tensor $\mathcal{F} = [[f_{p,(c,i)}]]$, find a $n,k$-complete perfect phylogeny tree $T$ and a $m,n,k$-usage matrix $U = [u_{p, (c,i)}]$
  such that $f_{p,(c,i)} = \sum_{(d,j) \in T_{(c,i)}} u_{p,(d,j)}$ for all character-state pairs $(c,i)$ and all samples $p$.
\end{problem}

In the two-state ($k=2$) case, $\mathcal{F}$ is a $2 \times m \times n$ tensor, and there is a convenient parameterization of the problem as $F = U A$, where $F$ is an $m \times n$ matrix containing the frequencies of the 1-states of each character in each sample and $A$ is a \emph{two-state} perfect phylogeny matrix.  Thus, the problem can be considered as a factorization problem where we are factoring $F$ into two matrices, each of which has a particular structure.  This is the Variant Allele Factorization problem given in \cite{ElKebir:2015by} and described in other forms elsewhere \cite{Nik-Zainal:2012aa,Strino:2013,Jiao:2014aa,Hajirasouliha:2014aa,Malikic:2015aa,Popic:2015,ElKebir:2015by}. When $k > 2$, the relationship between $\mathcal{F}$, $U$ and $A$, is given by $k$ systems of linear equations of the form $F_i = U A_i$, where $F_i$ and $A_i$ are ``slices" of $\mathcal{F}$ and a multi-state perfect phylogeny matrix $A$ for each of the $k$ states of the characters.
(Note that since $\sum_{i=0}^{k-1} f_{p,(c,i)} = 1$, one can instead write only $k-1$ systems of linear equations.)
See Supplementary Methods~\ref{sec:app_preliminaries} for further details and formal definitions of $\mathcal{F}$, $U$, $T$ and $A$.

\section{Combinatorial Characterization of the \Problem}
\label{sec:comb_char}
In this section, we derive a combinatorial characterization of the solutions of the \Problem\ as a restricted set of spanning trees in an edge-labeled, directed multi-graph.
We begin by reviewing and discussing previous results for mixtures of a two-state $(k=2)$ perfect phylogeny.
Next, we derive the characterization for mixtures of the multi-state ($k \ge 2$) perfect phylogeny.  Finally,  we discuss a special case of the multi-state perfect phylogeny problem using cladistic characters.
All proofs as well as additional definitions and lemmas are in Supplementary Methods~\ref{sec:app_methods}.

\subsection{Two-State Perfect Phylogeny Mixtures}
First, we review and recast the main results for mixtures of a two state ($k=2$) perfect phylogeny with all-zero root.
Here, a character changes state from 0 to 1 at most once.
The key insight is that the relative frequencies of the mutated ($=1$) states for a subset of characters constrain their potential ancestral relationships. This is because the mutated state persists in the tree.   In particular, if $(c,1) \prec_T (d,1)$ then all vertices in $T$ that have state 1 for character $d$ must also have state 1 for character $c$. 
A consequence is the following condition, called the \emph{ancestry condition} in \cite{ElKebir:2015by}:
\begin{equation}
  f_{p,(c,1)} \ge f_{p,(d,1)}~\text{for all samples $p$ and characters $c,d$ such that $(c,1) \prec_T (d,1)$}.
\tag{AC}
\end{equation}

In fact, a stronger condition than the ancestry condition can be derived by considering the relationships between subtrees of $T$.  Specifically, for each character $c$, the
subtree $T_{(c,1)}$, consisting of all vertices with state 1 for character $c$, \emph{is identical to} the subtree $\overline{T}_{(c,1)}$ rooted at a vertex $v_{(c,1)}$.
Moreover,  $\overline{T}_{(c,1)}$ is the disjoint union of $v_{(c,1)}$ and the subtrees rooted at its children: $\overline{T}_{(c,1)} =   v_{(c,1)} \cup \left(\bigcup_{(d,1) \in \delta(c,1)}  \overline{T}_{(d,1)} \right)$ (Figure~\ref{fig:t_vs_tbar}(a)).
Combining this fact together with the equation $f_{p,(c,1)} = \sum_{(d,1) \in T_{(c,1)}} u_{p,(d,1)}$ yields the key equation 
 $f_{p,(c,1)} =  u_{p,(c,1)} + \sum_{(d,1) \in \delta(c,1)} f_{p,(d,1)}$.
Recalling that  $u_{p,(c,1)} \ge 0$, we can relax this equation to the following inequality, referred to as the \emph{sum condition} in \cite{ElKebir:2015by},
\begin{equation}
  f_{p,(c,1)} \ge \sum_{(d,1) \in \delta(c,1)} f_{p,(d,1)}~\text{for all samples $p$ and characters $c$}.
\tag{SC}
 \label{eq:SC}
\end{equation}
The sum condition is both necessary and sufficient for $\mathcal{F}$ to be a mixture of $T$.
The  sum and ancestry conditions provide a combinatorial characterization of solutions as constrained spanning trees of a directed acyclic graph, which was called the \emph{ancestry graph} in \cite{ElKebir:2015by}. 
This derivation of the sum condition from the fact that $\overline{T}_{(c,1)} = T_{(c,1)}$ in the two-state case is not explicitly stated in previous work, but this turns out to be the key ingredient in the generalization to the multi-state case.

\subsection{Multi-State Perfect Phylogeny Mixtures}
We now consider a multi-state perfect phylogeny ($k \ge 2$).  In this case the state of a character $c$ can change more than once on the tree, but never changes back to a previous state. Thus in general $T_{(c,i)}$ \emph{is not equal} to $\overline{T}_{(c,i)}$ (Figure~\ref{fig:t_vs_tbar}(b)).
This makes the situation much more complicated, since we must consider not only the children of $v_{(c,i)}$, but also the relationships between $T_{(c,i)}$ and subtrees $T_{(c,j)}$ for $j \neq i$.  

\begin{figure}[t]
  \center
  \includegraphics[width=\textwidth]{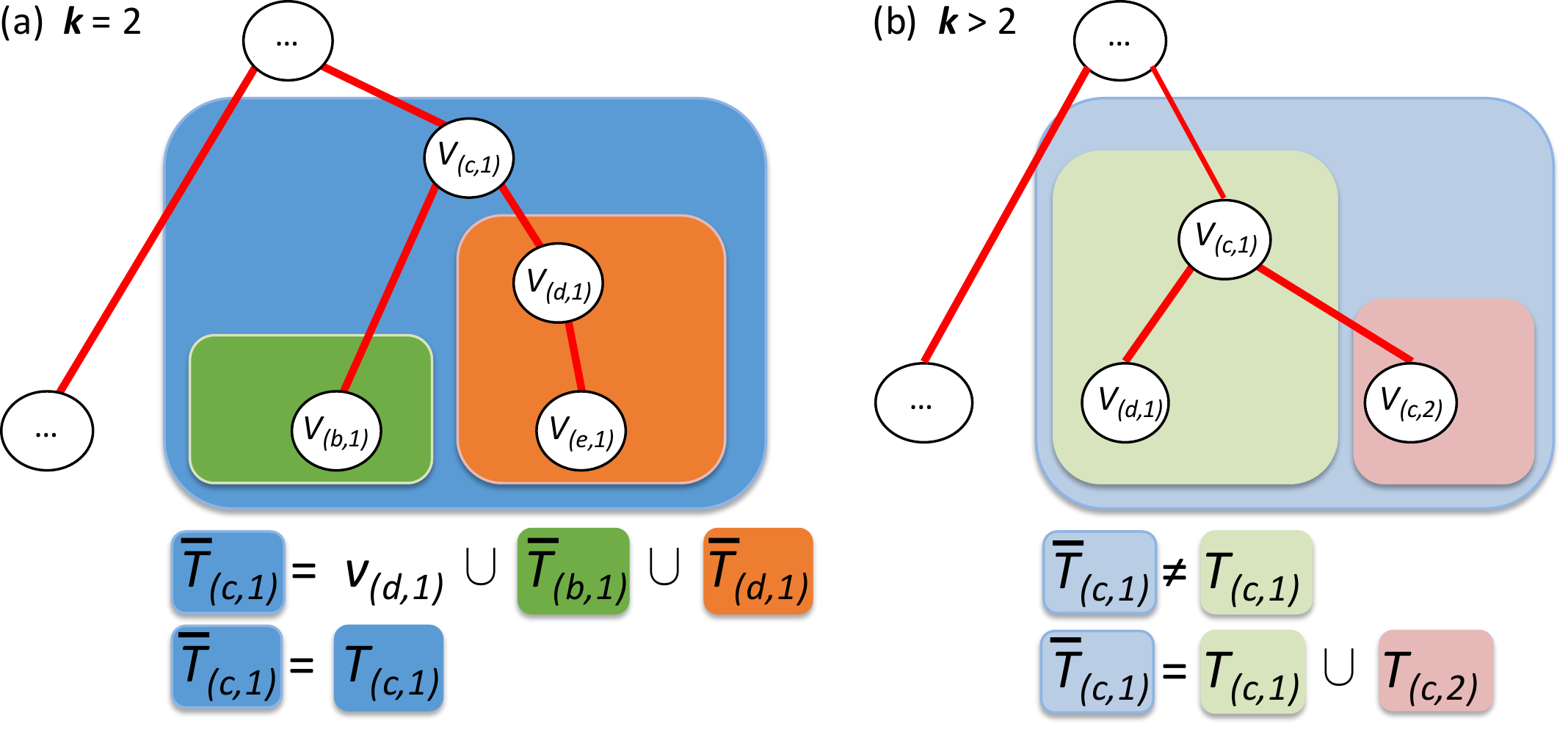}
  \caption{\textbf{Relationship between $T_{(c,i)}$ and $\overline{T}_{(c,i)}$ in the case of $k = 2$ and $k > 2$ states.} 
  (a) In the case of $k=2$ states, we have that $T_{(c,1)}) = \overline{T}_{(c,1)}$. 
  (b) In the case of $k > 2$ states, $T_{(c,i)} \neq \overline{T}_{(c,i)}$. 
  Instead, $\overline{T}_{(c,i)} = \bigcup_{l \in D_{(c,i)}} T_{(c,l)}$ where $D_{(c,i)}$ is the descendant set of $(c,i)$.
  Here, $D_{(c,1)} = \{1,2\}$.}
  \label{fig:t_vs_tbar}
\end{figure}

Our solution to the multi-state problem relies on the notion of a \emph{descendant set} for each character-state pair $(c,i)$.  Formally, 
given a complete perfect phylogeny tree $T$,  
we define the \emph{descendant set} $D_{(c,i)} = \{ j \mid (c,i) \prec_T (c,j) \}$ as the set of states for character $c$ that are descendants of character-state pair $(c,i)$ in $T$.
Note that $i \in D_{(c,i)}$, and that $D_{(c,0)} = \{0,\ldots,k-1\}$ as $v_{(c,0)} = v_{(*,0)}$.
The descendant set of a character precisely determines the relationship between $\overline{T}_{(c,i)}$ and $T_{(c,i)}$; namely we have $\overline{T}_{(c,i)} = \bigcup_{l \in D_{(c,i)}} T_{(c,l)}$. 
Hence, to obtain the usages of all vertices in $\overline{T}_{(c,i)}$ we must consider \emph{cumulative frequencies} $f^+_p(D_{(c,i)}) = \sum_{l \in D_{(c,i)}} f_{p,(c,l)}$.
Generalizing the result from the two-state case, we find that: given $T$,
the cumulative frequencies $f^+_p(D_{(c,i)})$ for the descendant sets defined by $T$ \emph{uniquely} determine the usage matrix $U$,
as shown in the following lemma.
\begin{lemma}
  \label{lem:U}
  Let $T \in \mathcal{T}_{n,k}$ and $\mathcal{F} = [[f_{p,(c,i)}]]$ be a frequency tensor.
  For a character-state pair $(c,i)$ and sample $p$, let
  \begin{equation}
    \label{eq:U}
    u_{p,(c,i)} = f^+_p(D_{(c,i)}) - \sum_{(d,j) \in \delta(c,i)} f^+_p(D_{(d,j)}).
  \end{equation}
  Then $U = [u_{p,(c,i)}]$ is the unique matrix whose entries satisfy $f_{p,(c,i)} = \sum_{(d,j) \in T_{(c,i)}} u_{p,(d,j)}$.
\end{lemma}

We say that $T$ \emph{generates} $\mathcal{F}$ if the corresponding matrix $U = [u_{p,(c,i)}]$, defined by \eqref{eq:U}, is a usage matrix, i.e.\  $u_{p, (c,i)} \ge 0$ for all $p$, $c$, and $i$, and $\sum_{(c,i)} u_{p, (c,i)} = 1$ for all $p$.
We now restate the problem as follows.
\begin{problem}[\Problem]
  Given an $k \times m \times n$ frequency tensor $\mathcal{F}$, does there exist a complete perfect phylogeny tree $T \in \mathcal{T}_{n,k}$ that generates $\mathcal{F}$? 
\end{problem}

It turns out that at that the positivity of values $u_{p,(c,i)}$ is a necessary and sufficient condition for $T$ to generate $\mathcal{F}$. 
This is captured by the Multi-State Sum Condition \eqref{eq:MSSC} in the following theorem.

\begin{theorem}
  \label{thm:characterization}
  A complete perfect phylogeny tree $T$ generates $\mathcal{F}$ if and only if 
  \begin{equation}
    \label{eq:MSSC}
    \tag{MSSC}
    f^+_p(D_{(c,i)}) - \sum_{(d,j) \in \delta(c,i)} f^+_p(D_{(d,j)}) \ge 0
  \end{equation}
  for all character-state pairs $(c,i)$ and all samples $p$.
\end{theorem}

Note that \eqref{eq:MSSC} is a generalization of the Sum Condition described in \cite{ElKebir:2015by}.  
In the two-state case, $D_{(c,1)} = \{1\}$ for all characters $c$, and thus the cumulative frequencies $f^+_{p}(D_{(c,1)})$ are equal to the input frequencies $f_{p,(c,1)}$. 
Thus, in the two-state case, the relative order of the frequencies of the mutated (=1) states of characters constrain the ancestral relationships between characters, as described above.  In contrast, in the multi-state case, we must consider the relative order of \emph{cumulative frequencies} $f^+_p(D_{(c,i)})$, but these depend on descendant sets which are \textit{a priori} unknown.  
We will show that for any tree $T$ that generates $\mathcal{F}$ and any pair $(c,i)$ and $(d,j)$ of character-state pairs with $(c,i) \prec_T (d,j)$, Theorem~\ref{thm:characterization} imposes a necessary condition on the corresponding descendant sets $D_{(c,i)}$ and $D_{(d,j)}$ determined by $T$.  
This condition is formalized in the following definition.

\begin{definition}
  \label{def:valid_state_set_pair}
  Let $(c,i)$ and $(d,j)$ be distinct character-state pairs and let $D_{(c,i)}, D_{(d,j)} \subseteq \{0,\ldots,k-1\}$.
  A pair $(D_{(c,i)}, D_{(d,j)})$ is a \emph{valid descendant set pair} provided 
  \begin{equation}
    \label{eq:MSAC}
    \tag{MSAC}
    f^+_p(D_{(c,i)}) - f^+_p(D_{(d,j)}) \ge 0
  \end{equation}
  for all samples $p$; and additionally if $c = d$ then $D_{(c,j)} \subsetneq D_{(c,i)}$.
\end{definition}

There are potentially many valid descendant set pairs as shown by the following lemma.
\begin{lemma}
  \label{lem:supersub}
  Let $(D_{(c,i)}, D_{(d,j)})$ be a valid descendant set pair. If $D_{(c,i)} \subseteq D'_{(c,i)}$ and $D'_{(d,j)} \subseteq D_{(d,j)}$ then $(D'_{(c,i)}, D'_{(d,j)})$ is a valid descendant set pair.
\end{lemma}

The Multi-State Ancestry Condition \eqref{eq:MSAC} is a generalization of the Ancestry Condition 
to $k \geq 2$. 
The following proposition shows that \eqref{eq:MSAC} is a necessary condition for solutions of the \Problem.

\begin{proposition}
  Let $T$ be a complete perfect phylogeny tree that generates $\mathcal{F}$. If $(c,i) \prec_T (d,j)$ then there exist a valid descendant set pair $(D_{(c,i)}, D_{(d,j)})$.
\end{proposition}

We now proceed to define the multi-state ancestry graph $G_\mathcal{F}$ whose edges correspond to valid descendant state pairs. This graph provides a combinatorial characterization of the solutions to the \Problem.

\begin{definition}
The \emph{multi-state ancestry graph} $G_{\mathcal{F}}$ of the frequency tensor $\mathcal{F}$ is an edge-labeled, directed multi-graph $G_{\mathcal{F}} = (V,E)$ whose vertices $v_{(c,i)}$ correspond to character-state pairs $(c,i)$ and whose multi-edges are $(v_{(c,i)},v_{(d,j)})$ for all valid descendant set pairs $(D_{(c,i)}, D_{(d,j)})$.
\end{definition}

Note that in the definition above, $v_{(1,0)},\ldots,v_{(n,0)}$ all refer to the same root vertex $v_{(*,0)}$.
In the $k=2$ case, $G_{\mathcal{F}}$ is a simple directed graph and 
the solutions to the \Problem\ are spanning trees of the ancestry graph that satisfy the sum condition \eqref{eq:SC}  \cite{ElKebir:2015by}.
When $k>2$, $G_{\mathcal{F}}$ becomes a directed edge-labeled multi-graph, and the labels of the multi-edges further constrain the set of allowed spanning trees.  We formalize this constraint by defining a \emph{threaded spanning tree} as follows.

\begin{definition}
  \label{def:threaded}
  A rooted subtree $T$ of $G_{\mathcal{F}} = (V,E)$ is a \emph{threaded tree} provided (1) for every pair of adjacent edges $(v_{(c,i)}, v_{(d,j)}), (v_{(d,j)}, v_{(e,l)}) \in E(T)$ with corresponding labels $(D_{(c,i)}, D_{(d,j)})$ and $(D'_{(d,j)}, D'_{(e,l)})$, it holds that $D_{(d,j)} = D'_{(d,j)}$, and (2) 
  for every pair of vertices $v_{(c,i)},v_{(c,j)} \in V(T)$ it holds that $D_{(c,j)} \subseteq D_{(c,i)}$ if and only if $(c,i) \prec_T (c,j)$.
\end{definition}
Threaded spanning trees of the multi-state ancestry graph $G_\mathcal{F}$ satisfying \eqref{eq:MSSC} are the solutions of the \Problem\ as stated in the following theorem.

\begin{theorem}
  \label{thm:combinatorial_characterization}
  A complete perfect phylogeny tree $T$ generates $\mathcal{F}$ if and only if $T$ is a threaded spanning tree of the multi-state ancestry graph $G_\mathcal{F}$  such that \eqref{eq:MSSC} holds.
\end{theorem}

In previous work \cite{ElKebir:2015by}, we derived a hardness result for the \Problem~in the case where $k=2$ and the number of samples $m = O(n)$.
Here, we prove a stronger hardness result where $m = 2$.
\begin{theorem}
  \Problem~is NP-complete even if $k=2$ and $m=2$.
\end{theorem}

\subsection{The~\ProblemCladisticFull}

A special case of the multi-state perfect phylogeny problem is the case of \emph{cladistic} multi-state characters~\cite{Fernandez:2001}, where we are given a set $\mathcal{S} = \{S_c \mid c \in [n]\}$ of \emph{state trees} for each character. A \emph{state tree} $S_c$ is a tree whose vertex set $\{0,\ldots,k-1\}$ are the states for character $c$, and whose edges describe the ancestral relationships between the states of character $c$.
We say that a perfect phylogeny $T$ is \emph{consistent} with $\mathcal{S}$ provided $(c,i) \prec_T (c,j)$ if and only if $i \prec_{S_{c}} j$   for all characters $c$ and states $i,j$.

The cladistic multi-state perfect phylogeny problem reduces to the binary case and is polynomial-time decidable~\cite{Fernandez:2001}. 
Thus, it is not surprising that the PPMDP also simplifies in the cladistic case.
In particular, the state tree $S_c$ determines the descendant state sets $D_{(c,i)}$ for each state $i$:  namely, \ $D_{(c,i)} = \{j \mid i \prec_{S_c} j \}$.
Therefore in the cladistic case, the multi-state ancestry graph becomes a simple graph with edges $(v_{(c,i)}, v_{(d,j)})$  labeled by $(D_{(c,i)}, D_{(d,j)})$, provided that \eqref{eq:MSAC} holds.
Moreover, since solutions of the \Problem\ have to be consistent with $S_c$, they will not contain edges $(v_{(c,i)}, v_{(c,j)})$ where $i$ is not the parent of $j$ in $S_c$ for each character $c$. 
We thus remove all such edges and arrive at the cladistic ancestry graph $G_{(\mathcal{F}, \mathcal{S})}$.
The following proposition formalizes the solutions in the cladistic case.
\begin{proposition}
A complete perfect phylogeny tree $T$ generates $\mathcal{F}$ and is consistent with state trees $\mathcal{S}$ if and only if $T$ is a threaded spanning tree of the cladistic ancestry graph $G_{(\mathcal{F},\mathcal{S})}$ such that \eqref{eq:MSSC} holds.
\end{proposition}

\section{Algorithm for the \ProblemCladistic}
In this section we describe an algorithm to find all threaded spanning trees in the cladistic ancestry graph $G_{(\mathcal{F},\mathcal{S})}$ that satisfy \eqref{eq:MSSC}.
We adapt the Gabow-Myers algorithm~\cite{Gabow:1978cg} for enumerating spanning trees to enumerate threaded spanning trees that satisfy \eqref{eq:MSSC}, following an approach used in \cite{Popic:2015} for the two-state problem with uncertain frequencies.
The crucial observation is that any subtree of a solution $T$ must also be a consistent, threaded tree and satisfy \eqref{eq:MSSC}.
Here, a subtree $T'$ is \emph{consistent} if it is rooted at $v_{(*,0)}$, and for each character $c$ the set of states $\{i \mid v_{(c,i)} \in V(T') \}$ induces a connected subtree in $S_c$.
Since the Gabow-Myers algorithm constructively grows spanning trees, we can arrive at the desired threaded spanning trees by maintaining the following invariant.

\begin{invariant}
  \label{inv:T}
  Let tree $T$ be the partially constructed tree.
  It holds that (1) for each $v_{(c,i)} \in V(T) \setminus \{v_{(*,0)}\}$ and parent $\pi(i)$ of $i$ in $S_c$, the vertex $v_{(c,\pi(i))}$ is the first vertex labeled by character $c$ on the unique path from $v_{(*,0)}$ to $v_{(c,i)}$; and (2) for each vertex $v_{(c,i)} \in V(T)$, \eqref{eq:MSSC} holds for $T$ and $\mathcal{F}$.
\end{invariant}

In addition, we maintain a subset of edges $H \subseteq E(G_{(\mathcal{F},\mathcal{S})})$ called the \emph{frontier} that can be used to extend $T$.  The requirement is that any edge $(v_{(c,i)}, v_{(d,j)}) \in H$ can be used to extend the partial tree $T$ without introducing a cycle or violating Invariant~\ref{inv:T}, which is captured by the following invariant.

\begin{invariant}
  \label{inv:H}
  Let tree $T$ be the partially constructed tree.
  For every edge $(v_{(c,i)},v_{(d,j)}) \in H$, (1) $v_{(c,i)} \in V(T)$, (2) $v_{(d,j)} \not \in V(T)$, and (3) Invariant~\ref{inv:T} holds for $T'$ where $E(T') = E(T) \cup \{(v_{(c,i)},v_{(d,j)})\}$.
\end{invariant}

Supplementary Algorithm~\ref{alg:cladistic} maintains the two invariants and reports all solutions to a \ProblemCladistic~instance $(\mathcal{F}, \mathcal{S})$.
The initial call is \textsc{Enumerate}$(G, \{v_{(*,0)}\}, \delta(*,0))$. Here, $\delta(*,0)$ corresponds to the set of outgoing edges from vertex $v_{(*,0)}$ of $G_{(\mathcal{F},\mathcal{S})}$, which by definition satisfies Invariant~\ref{inv:H}.

The running time is the same as the original Gabow-Myers algorithm: $O(|V| + |E| + |E| \times K)$ where $K$ is the number of spanning trees, disregarding \eqref{eq:MSSC}, in $G_{(\mathcal{F},\mathcal{S})} = (V,E)$.
In the case where $f_{p,(c,0)} = 1$ for all characters $c$ and samples $p$, the number of spanning trees of $G_{(\mathcal{F},\mathcal{S})}$ equals the number of threaded spanning trees that satisfy \eqref{eq:MSSC}.
Further details are in Supplementary Methods~\ref{sec:alg}.

In order to extend this algorithm to the general \Problem, we need to update the descendant sets $D_{(c,i)}$ as we grow the tree.
This in turn has implications for how we maintain the frontier: for each potential frontier edge, we need to consider how its addition to $T$ affects the descendant sets of the existing vertices of $T$ and thereby \eqref{eq:MSSC}.
We leave this extension as future work.

\subsubsection{Handling Errors in Frequency Tensor.}
In applications to real data, the frequencies $\mathcal{F} = [[f_{p,(c,i)}]]$ are often estimated from sequencing data and thus may have errors.
We account for errors in frequencies by taking as input nonempty intervals $[\bunderline{f}_{p,(c,i)}, \boverline{f}_{p,(c,i)}]$ that contain the true frequency $f_{p,(c,i)}$ for character-state pairs $(c,i)$ in samples $p$.
A tree $T$ is \emph{valid} if there exists a frequency tensor $\mathcal{F}' = [[f'_{p,(c,i)}]]$ such that $\bunderline{f}_{p,(c,i)} \le f'_{p,(c,i)} \le \boverline{f}_{p,(c,i)}$ and $\mathcal{F}'$ generates $T$---i.e.\ \eqref{eq:MSSC} holds for $T$ and $\mathcal{F}'$.
A valid tree $T$ is \emph{maximal} if there exists no valid supertree $T'$ of $T$, i.e.\ $E(T) \subsetneq E(T')$.
The task now becomes to find the set of all maximal valid trees.
Note that a maximal valid tree is not necessarily a spanning tree of $G_{(\mathcal{F},\mathcal{S})}$.

We start by recursively defining $\hat{\mathcal{F}} = [[\hat{f}_{p,(c,i)}]]$ where
\begin{equation}
  \hat{f}_{p,(c,i)} = \max\Big\{\bunderline{f}_{p,(c,i)}, \sum_{(d,j) \in \delta(c,i)} \hat{f}^+_p(D_{(d,j)}) - \sum_{j \in D(c,i) \setminus \{i\}} \hat{f}_{p,(c,j)}\Big\}.
\end{equation}
The intuition here is to satisfy \eqref{eq:MSSC} by assigning the smallest possible values to the children. We do this bottom-up from the leaves and set $\hat{f}_{p,(c,i)} = \bunderline{f}_{p,(c,i)}$ for each leaf vertex $v_{(c,i)}$.
It turns out that $\bunderline{f}_{p,(c,i)} \leq \hat{f}_{p,(c,i)} \leq \overline{f}_{p,(c,i)}$ is a necessary condition for $T$ to be valid as shown in the  following lemma.
\begin{lemma}
\label{lem:treevalid}
  If tree $T$ is valid then $\bunderline{f}_{p,(c,i)} \le \hat{f}_{p,(c,i)} \le \boverline{f}_{p,(c,i)}$ for all $p$ and $(c,i)$.
\end{lemma} 
The pseudo code of the resulting \textsc{NoisyEnumerate} algorithm and further details are given in Supplementary Methods~\ref{sec:alg}.

\section{Multi-State Model for Copy-Number Aberrations in Cancer Sequencing}
\label{sec:cancer_model}

Our motivating application for the PPMPD is to analyze multi-sample cancer sequencing data.  A tumor is a collection of cells that have evolved from the normal, founder cell by a series of somatic mutations.  Here, we assume that  $m$ distinct samples from a tumor have been sequenced, each such sample being a mixture of cells with different somatic mutations.  Our aim is to reconstruct the phylogenetic tree describing the evolutionary relationships between subpopulations of cells (or \emph{clones}) within the tumor.  Earlier work on this problem\cite{Jiao:2014aa,ElKebir:2015by,Malikic:2015aa} has focused on the problem of single-nucleotide variants (SNVs) that have changed at most once in the progression from normal to tumor, and thus the phylogenetic tree is a two-state perfect phylogeny.  Here we consider three additional types of mutations, copy-neutral loss-of-heterozygosity (CN-LOH), single-copy deletion (SCD) and single-copy amplification (SCA), that are common in tumors.
We do this by modeling a position in the genome as a multi-state character.

\begin{figure}[tbp]
  \centering
\includegraphics[width=0.35\textwidth]{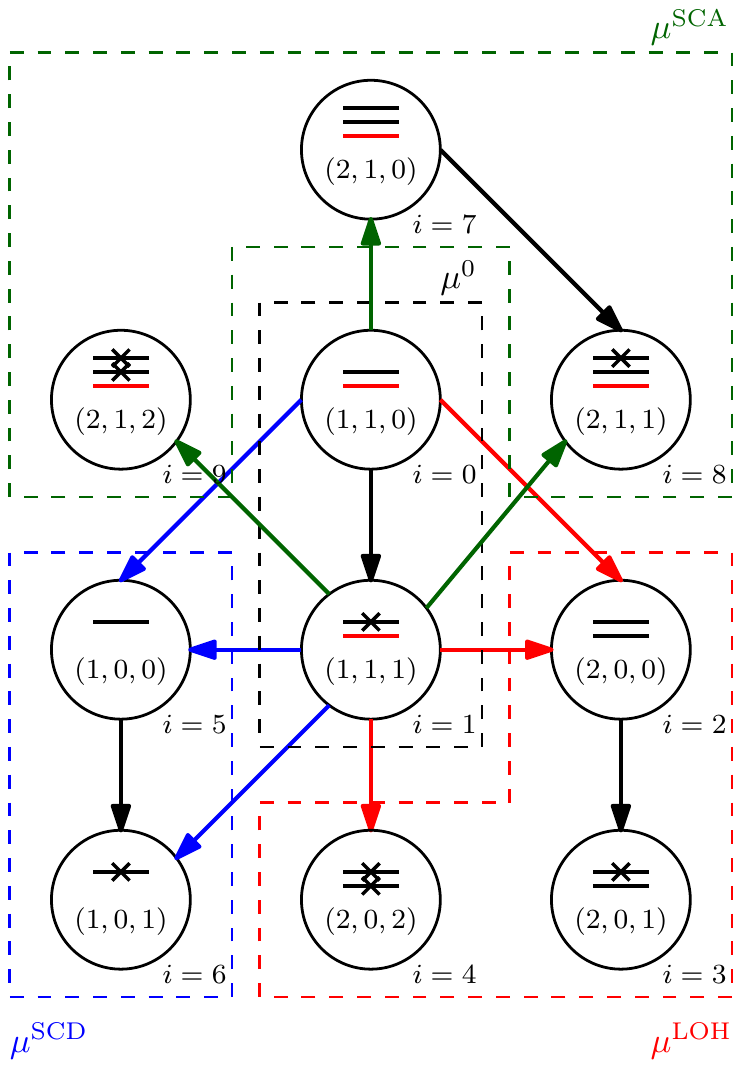}
\caption{\textbf{State graph for our model relating SNVs with CN-LOH, SCD and SCA events.} Edges are colored by the type of state transition:  black edges denote SNV events, the red edges denote CN-LOH events, the blue edges denote SCD events and the green edges denote SCA events.
  }
  \label{fig:CNA_SNV}
\end{figure}

The characters in our model are positions in the genome whose states we model with a triple $(x,y,z)$. Here,  $x$ and $y$ are the number of maternal and paternal copies of the position, respectively, and $z$ is the number of mutated copies.
We consider SNVs that are in regions that are unaffected by CNAs, or that have undergone CNA events that are copy-neutral loss-of-heterozygosity (CN-LOH), single-copy deletion (SCD) or single-copy amplification (SCA) events. A CN-LOH is an event where one chromosomal copy of a locus (either maternal or paternal) is lost and the remaining copy is duplicated so that the copy number of the locus remains 2 (diploid).  For each locus, we assume that the maternal copy number is at least as large as the paternal copy number, which, lacking phasing of chromosomes, does not constrain the model.
This results in the following ten states: non-mutated heterozygous diploid $(x_0, y_0, z_0) = (1,1,0)$, heterozygous diploid with SNV $(x_1, y_1, z_1) = (1,1,1)$, CN-LOH without SNV $(x_2, y_2, z_2) = (2,0,0)$, CN-LOH prior to SNV $(x_3, y_3, z_3) = (2,0,1)$, CN-LOH retaining SNV $(x_4, y_4, z_4) = (2,0,2)$, SCD without SNV $(x_5, y_5, z_5) = (1,0,0)$, SCD retaining SNV $(x_6, y_6, z_6) = (1,0,1)$, SCA without SNV $(x_7,y_7,z_7) = (2,1,0)$, SCA prior to SNV $(x_8,y_8,z_8) = (2,1,1)$ and SCA of SNV $(x_9,y_9,z_9) = (2,1,2)$.

Unfortunately, for a character (genomic locus) $c$ and sample $p$, we cannot derive the frequencies $f_{p,(c,i)}$ for states $i \in \{0,\ldots,9\}$ directly from DNA sequencing data.  
Instead, alignment of DNA sequence reads from a tumor and matched normal genome gives us measurements of three quantities: \emph{variant allele frequencies} for SNVs, which is the proportion of sequence reads that contain the mutant allele;  \emph{B-allele frequencies} for germline single-nucleotide polymorphisms which reveal LOH events; \emph{read-depth ratios} for larger regions, which are the total number of the reads of the region in the tumor sample divided by the total number of reads in the same region of a matched normal~\cite{Raphael:2014aa}.  We will show that with a few simplifying assumptions about the allowed state transitions, we can derive the frequencies $f_{p,(c,i)}$ from these three quantities as well as the state trees for each character.  This will give us an instance of \ProblemCladistic.

We define a \emph{state graph} to be a directed graph whose nodes are character states, and whose edges are allowed transitions between states (Figure~\ref{fig:CNA_SNV}).  Because the ancestor of all cells in the tumor is a normal cell, the root state is $(1,1,0)$.
We assume that a locus undergoes one SNV event and at most one CN-LOH, SCD or SCA event.
Thus we derive state trees for each character from the state graph by considering rooted subtrees of the graph that use one SNV edge and at most one CN-LOH, SCD or SCA edge.
This results in the set $\Sigma$ comprised of 13 possible state trees, which are shown in Supplementary~Figure~\ref{fig:LOH_SCD_SNV}.

We denote the variant allele frequency (VAF) of character $c$ in sample $p$ by $h_{p,c}$.
Using the read-depth ratio and BAFs, we determine whether a locus $c$ in sample $p$ has undergone a CN-LOH or an SCD event, by using the 
THetA~\cite{oesper2013theta,Oesper:2014aa} algorithm to cluster read-depth ratios and BAFs.
For each locus $c$ and sample $p$, THetA also determines the proportions:  $\mu^\mathrm{LOH}_{p,c}$, $\mu^\mathrm{SCD}_{p,c}$, $\mu^\mathrm{SCA}_{p,c}$, $\mu^0_{p,c}$ which are the proportion of cells that have undergone an CN-LOH, SCD, SCA event, or are unaffected by a copy-number aberration (CNA), respectively at locus $c$  in sample $p$. 
By the definition of the states, we have that $\mu^0_{p,c} = f_{p,(c,0)} + f_{p,(c,1)}$, $\mu^\mathrm{LOH}_{p,c} =  f_{p,(c,2)} + f_{p,(c,3)} + f_{p,(c,4)}$, $\mu^\mathrm{SCD}_{p,c} =  f_{p,(c,5)} + f_{p,(c,6)}$ and $\mu^\mathrm{SCA}_{p,c} =  f_{p,(c,7)} + f_{p,(c,8)} + f_{p,(c,9)}$ (Figure~\ref{fig:CNA_SNV}). The following equation relates the VAF $h_{p,c}$ to frequencies $f_{p,(c,i)}$ and thus to proportions $\mu^0_{p,c}$, $\mu^\mathrm{LOH}_{p,c},\mu^\mathrm{SCD}_{p,c}$ and $\mu^\mathrm{SCA}_{p,c}$.
\begin{equation}
  \label{eq:vaf2}
  h_{p,c} =  \frac{\sum_{i=0}^9 z_i \cdot f_{p,(c,i)}}{\sum_{i=0}^9 (x_i + y_i) \cdot f_{p,(c,i)}}
\end{equation}

According to our assumption, we have that for each character $c$ across all samples $p$ at most one of the proportions $\mu^\mathrm{LOH}_{p,c}, \mu^\mathrm{SCD}_{p,c}, \mu^\mathrm{SCA}_{p,c}$ is nonzero.
Moreover, by selecting a state tree $S_c \in \Sigma$, we fix the frequencies of the absent states of $S_c$ to 0.
We thus have a linear system of equations with variables $f_{p,(c,i)}$ and constants $h_{p,c}$, $\mu^0_{p,c}$, $\mu^\mathrm{LOH}_{p,c}$, $\mu^\mathrm{SCD}_{p,c}$ and $\mu^\mathrm{SCA}_{p,c}$. 
Solving this system results in a unique solution for each variable $f_{p,(c,i)}$ (Supplementary Table~\ref{tab:CNA}).
Hence for each state tree $S_c$, we obtain frequencies $f_{p,(c,i)}$. 
If all the frequencies are nonnegative, we say that state tree $S_c$ is \emph{compatible} with $c$. 
By enumerating combinations of compatible state trees, we arrive at instances $(\mathcal{F},\mathcal{S})$ of the \ProblemCladistic.
We account for noise in VAFs $h_{p,c}$ by considering their confidence intervals $[\bunderline{h}_{p,c},\boverline{h}_{p,c}]$. 
In a similar fashion to the above, we derive frequency intervals $[\bunderline{f}_{p,(c,i)}, \boverline{f}_{p,(c,i)}]$ from the selected state tree $S_c \in \Sigma$, the VAF confidence interval $[\bunderline{h}_{p,c},\boverline{h}_{p,c}$], and the proportions $\mu^0_{p,c}$, $\mu^\mathrm{LOH}_{p,c}$, $\mu^\mathrm{SCD}_{p,c}$ and $\mu^\mathrm{SCA}_{p,c}$ (Supplementary Table~\ref{tbl:intervals}).

\section{Results}

We apply our multi-state perfect phylogeny model for copy-number aberrations to analyze 180 simulated datasets, a chronic lymphocytic leukemia tumor from~\cite{Schuh:2012aa} and a prostate cancer tumor from~\cite{Gundem:2015}.

\subsection{Error-Free Simulations}  
We create $60$ multi-state perfect phylogeny trees $T$ containing $n \in \{4,5,6\}$ characters, using state trees from $\Sigma$ as defined  in the previous section.  For each tree $T$ and number of samples $m \in \{2,5,10\}$ we simulate a frequency tensor $\mathcal{F}$ by mixing vertices from $T$, resulting in 180 simulated mixtures. 
Next, we use the entries of the simulated $\mathcal{F}$ to generate the actual input by computing for each character $c$ in sample $p$, the VAF $h_{p,c}$ and the proportions $\mu^0_{p,c}$, $\mu^\mathrm{LOH}_{p,c}$, $\mu^\mathrm{SCD}_{p,c}$ and $\mu^\mathrm{SCA}_{p,c}$---as defined in Section~\ref{sec:cancer_model}. 
We run \textsc{Enumerate} for each combination of compatible state trees and only retain the largest maximal valid perfect phylogeny trees.
Figure~\ref{fig:res_A} shows that 
as the number of characters increases, the size of the solution space also increases, but that
increasing the number of samples reduces the number of solutions by several orders of magnitude. For example, while there are on average 34,586 solutions for $n=6$ characters with $m=2$ samples, there are on average only 7 solutions with $m=10$ samples for the same number of characters.  This demonstrates how the \eqref{eq:MSSC} becomes a strong constraint with increasing numbers of samples.   Figure~\ref{fig:res_B} shows a summary of the 45 solutions for one input instance ($n=6$, $m=5$).

We define the \emph{concordance} of a tree in the solution space to be the fraction of edges in the simulated tree that are recovered in the solution. Because we are enumerating the full solution space, we are guaranteed to find a solution with a concordance of 1 (i.e.\ the true tree). The spread of the distribution of concordances for the solution space is a measure of the level of ambiguity in the data.
Figure~\ref{fig:res_C} shows the distributions of these measurements across different number of samples for the previously considered tree. 
We see that an increase in the number of samples corresponds with an increase in the concordance of the solution space. 
In summary, we can deal with ambiguity by increasing the number of samples, which will decrease the size of solution space while at the same time increasing the concordance.

\subsection{Simulations with VAF Errors}

We now consider how well we can reconstruct the 180 simulated mixtures described above in the presence of errors in the VAFs. For each character $c$ and sample $p$, we draw its total read count from a Poisson distribution parameterized by a target coverage.
Next we draw the number of variant reads from a binomial parameterized by the previously drawn total read count and the true VAF.
We then consider the posterior distribution of observing the drawn total and variant read counts from which we compute a 0.95 confidence intervals on the VAF using a beta posterior~\cite{ElKebir:2015by}.
We run \textsc{NoisyEnumerate} on the simulated instances on each combination of compatible state trees, outputting the largest trees.
Supplemental Figure~\ref{fig:res_D} shows the size of the solution space for different values of the target coverage for all 20 instances with $n=4$ characters. 
Note that increasing the coverage, results in a decrease of the size of the solution space. In fact, a coverage of 10,000x approaches the error-free data. However, a more efficient way to deal with ambiguity is to increase the number of samples. For instance, a target coverage of 50x with 5 samples has a similar number of solutions as a coverage 10,000x with 2 samples.
We observe the same trends with $n=5$ characters and also in terms of running time (Supplementary Figure~\ref{fig:noisy}).

As the number of characters increases, exhaustive enumeration of the solution space becomes infeasible.  Thus, we analyzed how well we reconstruct the true tree using a fixed number $N$, of maximal trees.  Supplemental Figure~\ref{fig:res_E} shows the concordance of the solutions for a simulated instance with $n=10$ characters and a target coverage of 1,000X.

\begin{figure}
  \centering
 \subfloat[\label{fig:res_A}]{\includegraphics[width=.40\textwidth]{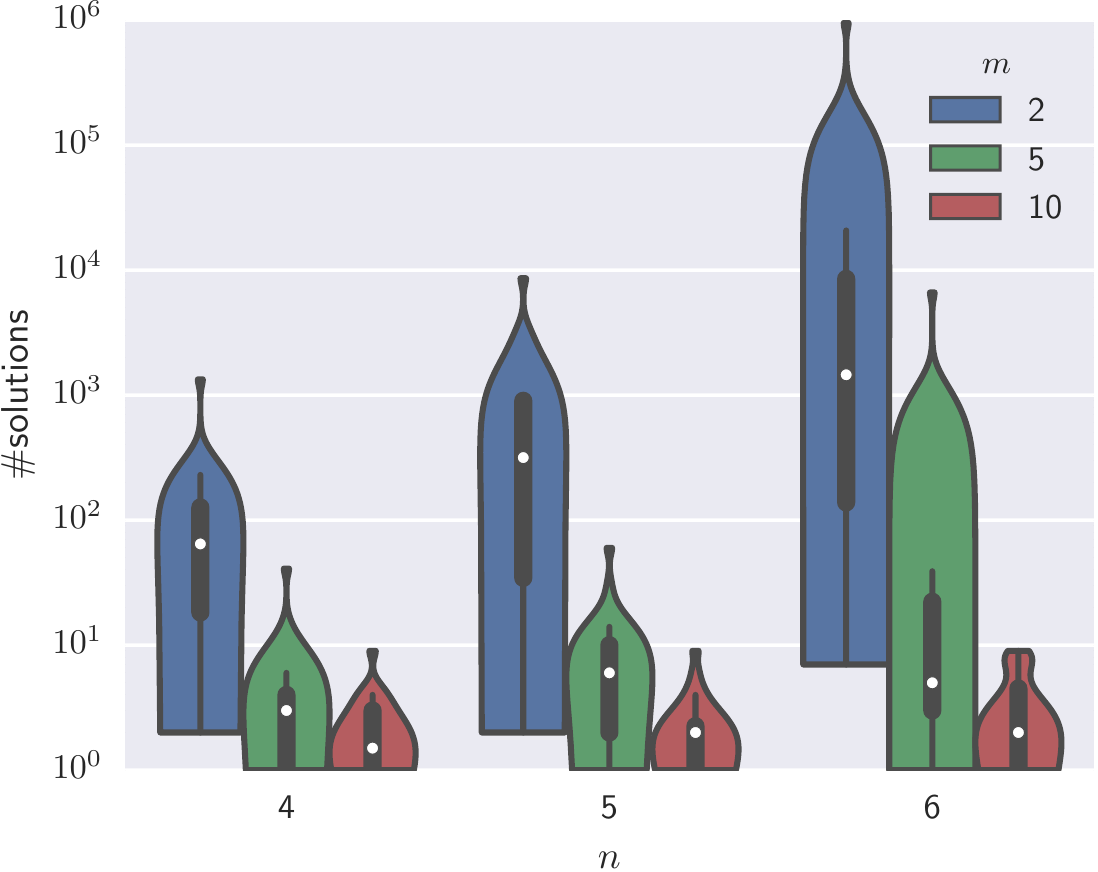}}
  \hspace{.05cm}
  \subfloat[\label{fig:res_B}
  ]{\includegraphics[width=.32\textwidth]{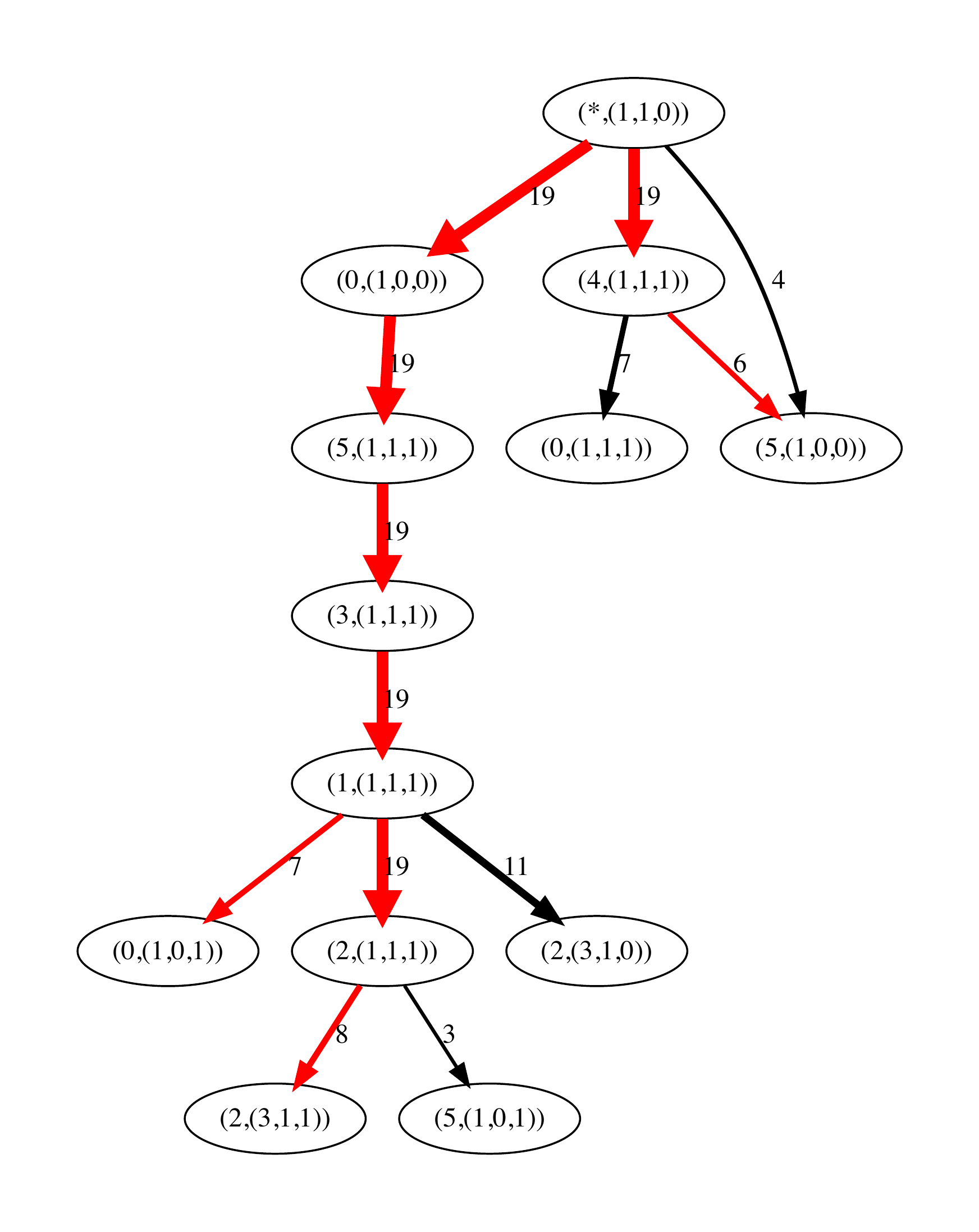}}
  \hspace{.05cm}
  \subfloat[\label{fig:res_C}]{\includegraphics[width=.20\textwidth]{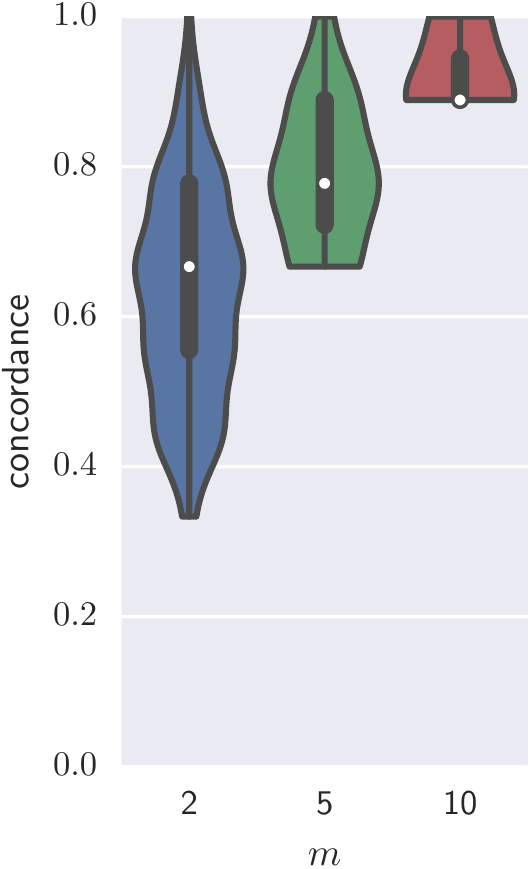}}
  \caption{\textbf{Simulated data results.} (a) Number of solutions (log-scale) for error-free data with $n \in \{4,5,6\}$ characters and $m \in \{2,5,10\}$ samples. (b) Solution space of an $n=6$, $m=5$ instance. Vertices are from the solution trees and each edge is labeled by the number of solutions in which it occurs. Red edges indicate the simulated tree. (c) Concordance of one $n=6$ instance. 
  The green violin corresponds to Figure~\ref{fig:res_B}.
  }
\end{figure}

\subsection{Real Data}
\label{sec:real_data}

We next apply the algorithm on a liquid chronic lymphocytic leukemia tumor (CLL077) from \cite{Schuh:2012aa} and  a solid prostate cancer tumor (A22) from \cite{Gundem:2015}. 

\subsubsection{Results on CLL077.}

We used targeted and whole-genome sequencing data from four time-separated samples (b, c, d, e).
The targeted data includes 14 SNVs, one of which (\textit{SAMHD1}) is classified as an CN-LOH in all four samples.
Two SNVs (in genes \textit{BCL2CB} and \textit{NAMPTL})  were classified as being unaffected by CNAs, but in some of the samples they had had a VAF confidence interval greater than 0.5 and as such were incompatible with all state trees.
The 12 remaining characters had only one compatible state tree associated with them.
We ran \textsc{NoisyEnumerate} until completion, and thus enumerated the entire solution space, which consists of 20 trees of nine vertices (Supplementary Figure~\ref{fig:solution_space_cll}).
Figure~\ref{fig:CLL_tree} shows one tree from the solution space. A similar tree with two branches is also reported by PhyloSub~\cite{Jiao:2014aa}, PhyloWGS~\cite{Deshwar:2015kw}, CITUP~\cite{Malikic:2015aa} and AncesTree~\cite{ElKebir:2015by} for this dataset.  However, the tree reported here and the one reported by AncesTree predict the order of all the mutations on each branch, while PhyloSub, PhyloWGS and CITUP group some mutations together.  Additionally, AncesTree did not consider the SNV in gene \textit{SAMHD1}, as its VAF~$> 0.5$.
Here, we reconstruct a tree containing the CN-LOH event on \textit{SAMHD1}.

By enumerating the entire search space, we can detect ambiguities in the input data.
For instance, in our tree \emph{LRRC16A} is a child of \emph{EXOC6B} whereas there are solutions which assign \emph{LRRC16A} as a child of either \emph{OCA2} or \emph{DAZAP1} (which is absent in the shown tree).
Without additional data or further assumptions, there is not enough information to distinguish between these ancestral relationships.
In contrast, by only providing one solution, AncesTree and CITUP give an incomplete picture that does not reflect the true uncertainty inherent to the data.

\begin{figure}[h]
  \center
  \subfloat[Tree for CLL077]{\label{fig:CLL_tree}
  \includegraphics[width=.3\textwidth]{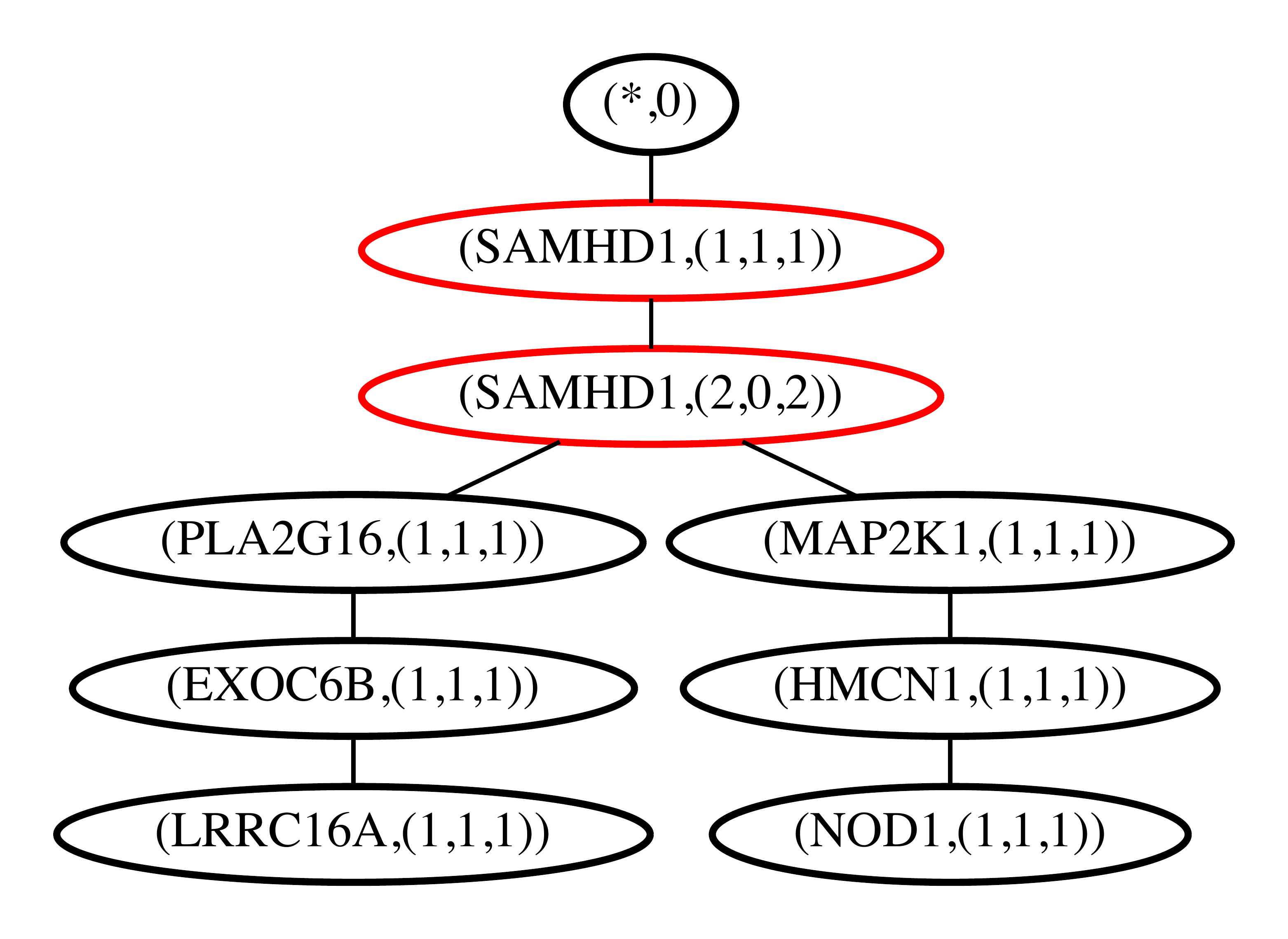}}
  \subfloat[\label{fig:A22_tree}Tree for A22]{\includegraphics[width=.69\textwidth]{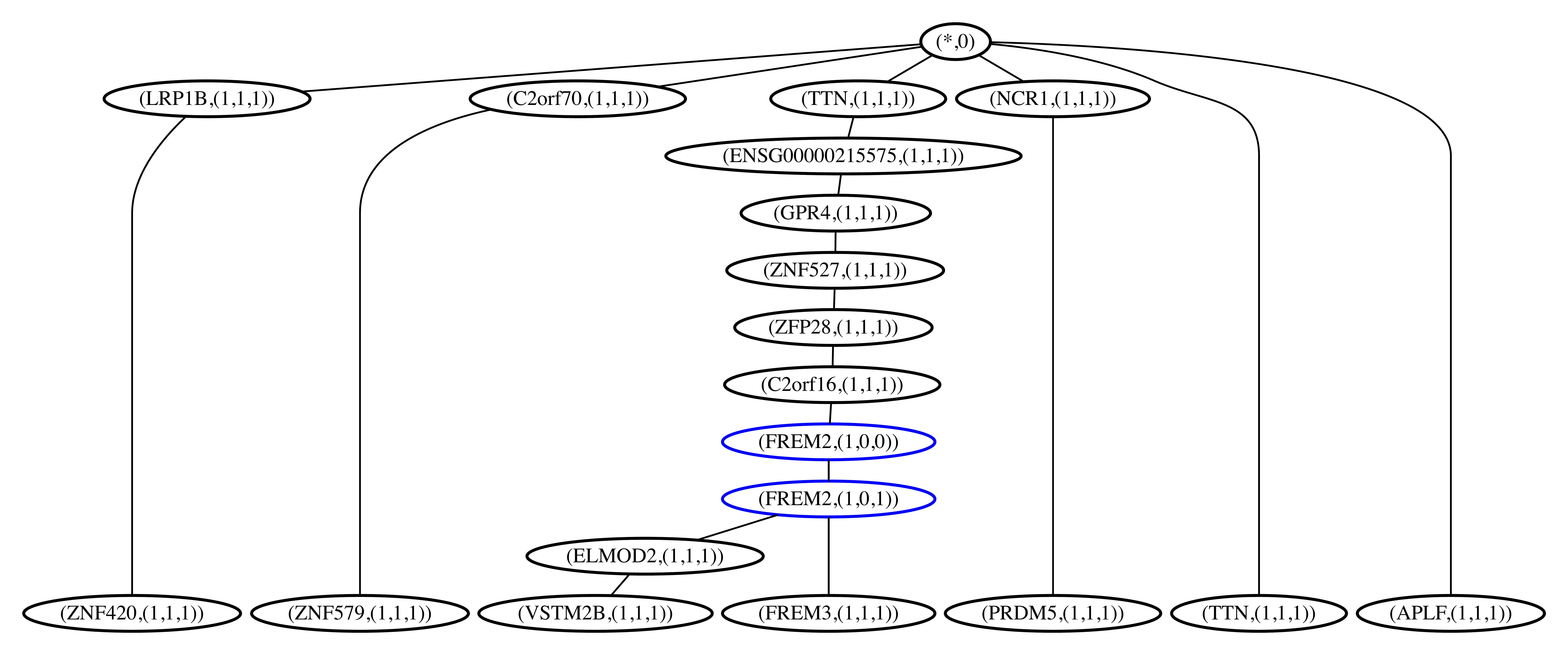}}
  \\
  \subfloat[\label{fig:CLL077_U}Usage matrix for CLL077]{\includegraphics[width=.3\textwidth]{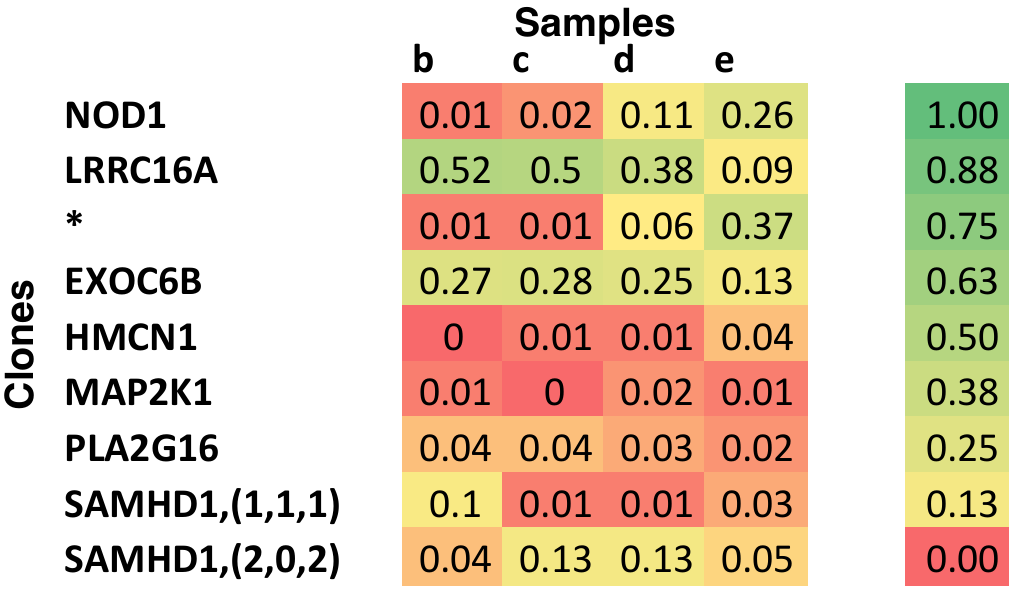}}
  \hspace{1cm}
  \subfloat[\label{fig:A22_U}Usage matrix for A22]{\includegraphics[width=.38\textwidth]{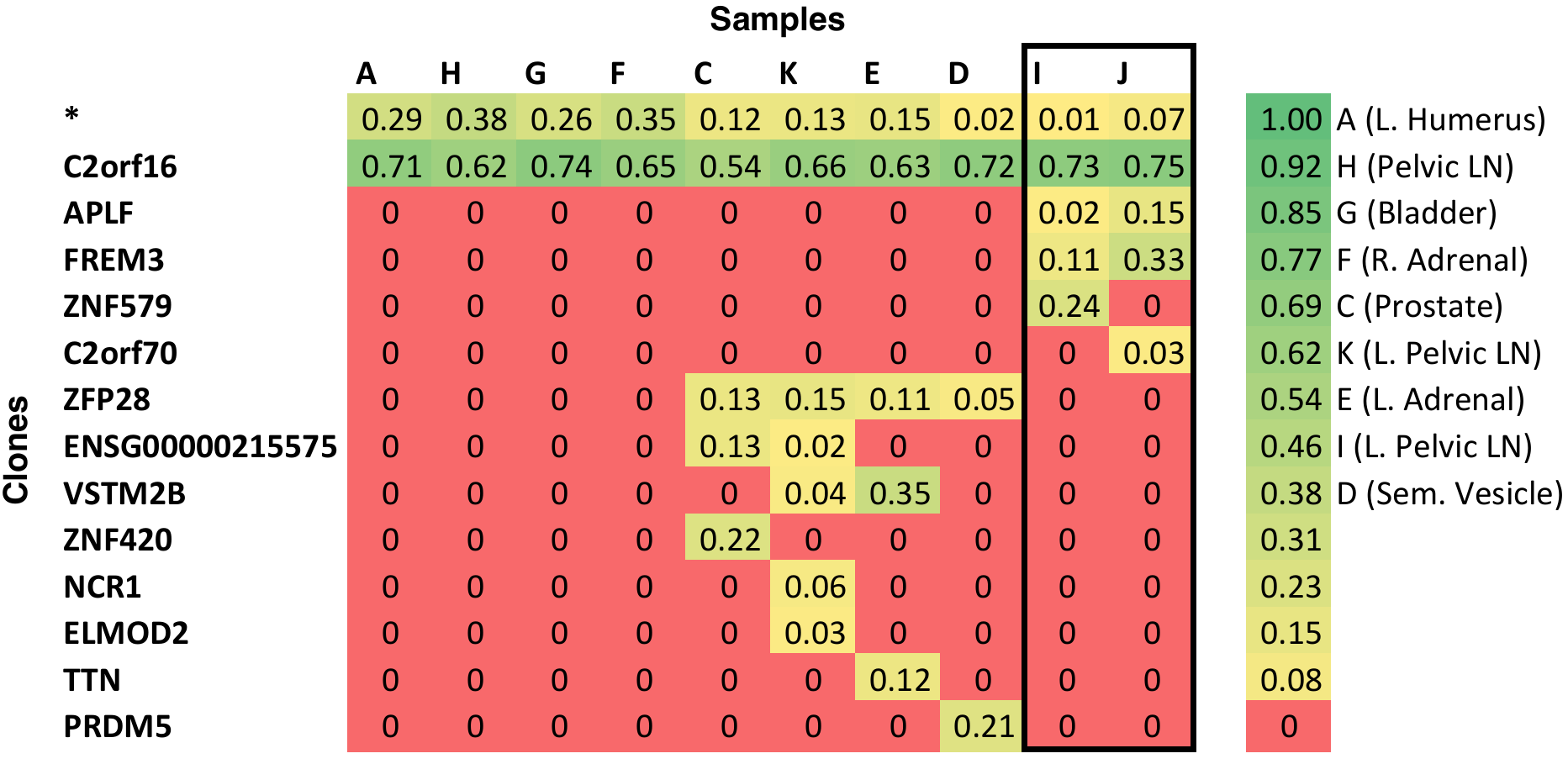}}
  \caption{\textbf{Real data results.} (a) The computed tree for CLL077 with a CN-LOH event in red. (b) The computed tree for A22 with a SCD event in blue.  (c) Usage matrix for CLL077 shows that the samples (columns) are mixed and consist of many clones (rows) as indicated by the coloring, (d) Usage matrix for A22 shows that samples consist of small subsets of clones, which reflect their distinct spatial locations}
\end{figure}

\subsubsection{Results on A22.}
Next, we consider a solid prostate cancer tumor (A22)~\cite{Gundem:2015} where 10 samples were taken from the primary tumor and different metastases.
The number of SNVs is 114.
Applying THetA showed that this tumor is highly rearranged. 
We consider only SNVs that are in regions classified as CN-LOH or SCD across all samples and whose VAFs are greater than 0.01 in all samples.
This resulted in a set of 27 SNVs.

We restrict the enumeration to $N=10^6$ maximal trees.
\textsc{NoisyEnumerate} finds 24,288 solutions comprised of 20 vertices (Supplementary Figure~\ref{fig:solution_space_prostate}).
Figure~\ref{fig:A22_tree} shows a representative tree of the solution space, i.e.\ the solution tree that shares the largest number of edges with other trees in the solution space. This tree has a SCD event containing gene \textit{FREM2}, which has a VAF $> 0.5$  in 8 of 10 samples. Since a VAF $> 0.5$ for an SNV violates the assumption of two-state perfect phylogeny, methods that use this assumption will disregard this locus.
In the inferred tree, the parent of \textit{FREM2} is \textit{C2orf16}, but the VAF of the SNV in this gene is lower than \textit{FREM2} in every sample. 
Thus, the VAFs of SNVs in isolation provide insufficient evidence to infer the ancestral relationship between \textit{FREM2} and \textit{C2orf16}, whereas combining the VAFs with BAFs and read-depth ratios allows us to do so.

Figure~\ref{fig:A22_U} shows the usage matrix for this solution.
In contrast to the CLL tumor, we do not expect the clones to be well mixed, since the primary tumor is a solid tumor and the metastases samples are physically separated from the primary tumor.
Indeed, we find clones that are specific to certain samples and that there is no sample consisting of all clones.
In addition, we see that certain samples are more similar to each other in terms of their usages. 
In particular, samples $I$ and $J$ only differ in two clones and both correspond to pelvic lymph nodes.
In summary, we find that the samples consist of small subsets of clones that reflect that they correspond to distinct spatial locations of the samples.

\section{Discussion}
We introduce the \ProblemFull\ for multi-state characters.  We describe both a combinatorial characterization of the solutions to this problem, and an algorithm to solve it in the presence of measurement errors.  Using this algorithm, we find that even with a small number of characters, there is extensive ambiguity in the solution with a modest number of samples ($m \approx 5$), but this ambiguity declines substantially as the number of samples increases ($m \approx 10$).  We analyze two tumor datasets and find that analysis of both copy-number aberrations and SNVs is required to obtain accurate phylogenetic trees.

The combinatorial structure derived here for the multi-state problem could be useful for the development of better probabilistic models, which have proved useful both in clustering mutations~\cite{Nik-Zainal:2012aa} and in simultaneous clustering and tree inference \cite{Jiao:2014aa,Deshwar:2015kw,MarkowetzANdBeerenwinkel}.  For example, rather than considering generic tree-structured priors, one could use priors that are informed by the combinatorial structure of the multi-state ancestry graph.

Although we focused on applications to cancer genome sequencing, the algorithm has applications in other cases  of mixed samples, including metagenomics~\cite{kembel2011, segata2013} and studying the process of somatic hypermutation. The latter was explored by~\cite{Strino:2013} using a single-sample, two-state perfect phylogeny model.  Some of these applications, as well as the cancer application, may require further relaxation of the infinite alleles model that we used here.  It is an interesting question whether more complicated phylogenetic models (e.g.\ the maximum parsimony model or more complicated copy number models~\cite{Chowdhury:2015kd}) can be analyzed in the setting of phylogenetic mixtures.


\bibliographystyle{plain}
\bibliography{k-vaffp}

\clearpage
\appendix

\section{Supplementary Methods}
\label{sec:app_methods}
In this section we present additional definitions and results.  Proofs of theorems, lemmas and propositions that were omitted in the main text are marked as such and maintain the same numbering as in the main text.

\label{sec:appendix_proofs}

\renewcommand\thefigure{\thesection\arabic{figure}}
\setcounter{figure}{0}

\subsection{The \ProblemFull}
\label{sec:app_preliminaries}


Recall that $m$ is the number of samples, $n$ is the number of characters and $k$ is the number of states of each character.
Our input measurements are given by the $k \times m \times n$ frequency tensor $\mathcal{F} = [F_i] = [ [ f_{p,(c,i)} ] ]$ where $f_{p,(c,i)}$ is the proportion of taxa of sample $p$ that have state $i$ for character $c$. We denote by $F_i$ the slice of $\mathcal{F}$ where the state of each character is $i$. 
Formally, $\mathcal{F}$ is defined as follows.
\begin{definition}
  \label{def:F}
  An $k \times m \times n$ tensor $\mathcal{F} = [F_i] = [[f_{p,(c,i)}]]$ is a \emph{frequency tensor} provided $f_{p,(c,i)} \geq 0$ and $\sum_{i = 0}^{k-1} f_{p,(c,i)} = 1$ for all characters $c$ and samples $p$.
\end{definition}

As mentioned in Section~\ref{sec:problem_statement}, the goal is to explain the observed frequencies $\mathcal{F}$ as $m$ mixtures of the leaves of a perfect phylogeny tree $T$, where each mixture corresponds to one sample. We recall the definition of a perfect phylogeny\cite{gusfield2014recombinatorics,Fernandez:2001}.
\begin{definition}
A rooted tree $T$ is a \emph{perfect phylogeny tree} provided that (1) each vertex is labeled by a \emph{state vector} in $\{0,\ldots,k-1\}^n$, which denotes the state for each character; (2) the root vertex of $T$ has state 0 for each character; (3) vertices labeled with state $i$ for character $c$ form a connected subtree $T_{(c,i)}$ of $T$.
\end{definition}
Rather than explaining $\mathcal{F}$ as mixtures of the leaves of a perfect phylogeny tree, we aim to explain $\mathcal{F}$ as $m$ mixtures of \emph{all} vertices of an $n,k$-complete perfect phylogeny tree, which is defined as follows.

\begin{definition}
  \label{def:T}
  An edge-labeled rooted tree $T$ on $n(k-1) + 1$ vertices is a \emph{$n,k$-complete perfect phylogeny tree} provided each of the $n(k-1)$ edges is labeled with exactly one character-state pair from $[n] \times [k-1]$ and no character-state pair appears more than once in $T$. Let $\mathcal{T}_{n,k}$ be the set of all $n,k$-complete perfect phylogeny trees.
\end{definition}

We may do this without loss of generality, as each $n,k$-complete perfect phylogeny $T$ can be mapped to a perfect phylogeny tree $T'$ by extending inner vertices  of $T$ that have non-zero mixing proportions to leaves of $T'$. See Supplementary Figure~\ref{fig:mapping} for an example.
In the following we denote by $T$ an $n,k$-complete perfect phylogeny tree $T$, by $v_{(c,i)}$ the vertex of $T$ whose incoming edge is labeled by $(c,i)$, and by $v_{(*,i)}$ the root of $T$.
Alternatively, $v_{(1,0)},\ldots, v_{(n,0)}$ all refer to the root vertex $v_{(*,0)}$.

\begin{figure}[h]
  \center
  \includegraphics[width=\textwidth]{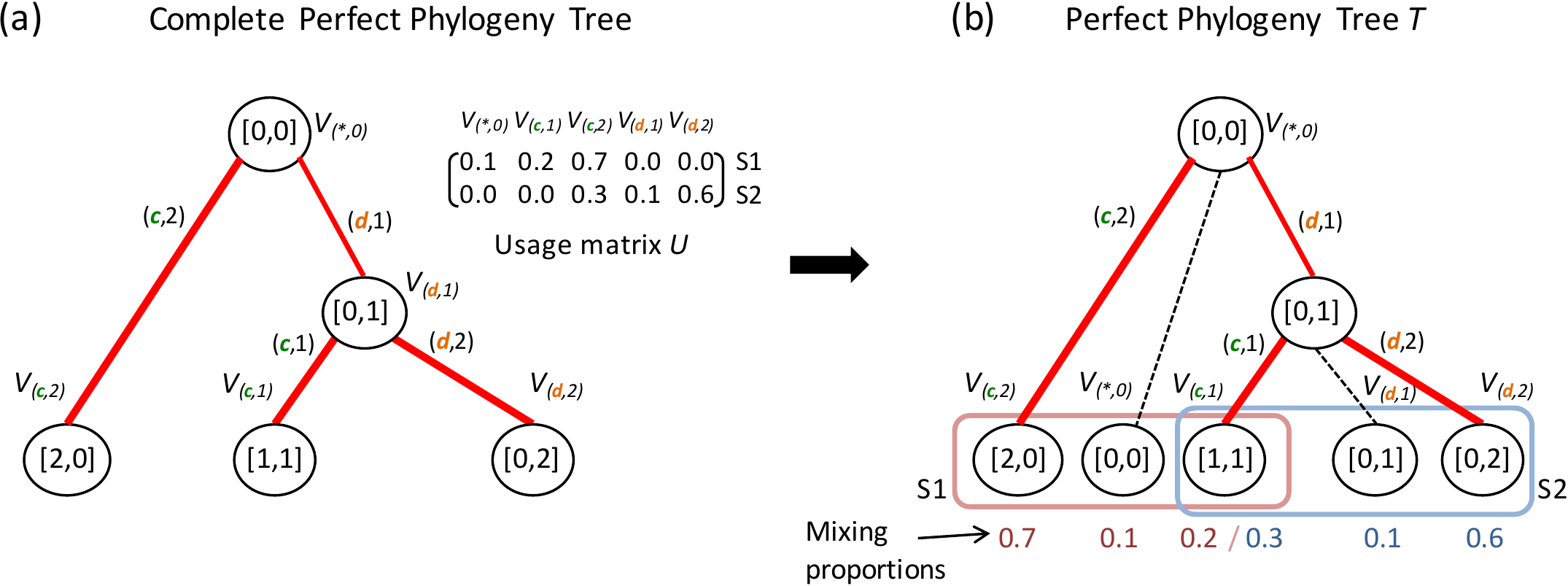}
  \caption{\textbf{Each complete perfect phylogeny tree corresponds to a perfect phylogeny tree}}
  \label{fig:mapping}
\end{figure}

The observed frequencies $\mathcal{F}$ are related to the vertices of $T$ by an $m \times (n(k-1) + 1)$ matrix $U = [u_{p,(c,i)}]$ whose rows correspond to samples and columns to vertices of $T$ such that each entry $u_{p,(c,i)}$ indicates the mixing proportion of the vertex $v_{(c,i)}$ of $T$ in sample $p$.
More specifically, each frequency $f_{p,(c,i)}$ is the sum of mixing proportions of all vertices of $T$ that possess state $i$ for character $c$, i.e.\
$f_{p,(c,i)} = \sum_{(d,j) \in T_{(c,i)}} u_{p,(d,j)}$ (Supplementary Figure~\ref{fig:F_U_relation}).
Formally, we define a usage matrix $U$ as follows.
\begin{definition}
  An $m \times (n(k-1)+1)$ matrix $U = [u_{p,(c,i)}]$ is an \emph{$m,n,k$-usage matrix} provided $u_{p,(c,i)} \ge 0$, and $u_{p,(*,0)} + \sum_{c = 1}^n \sum_{i = 1}^{k-1} u_{p,(c,i)} = 1$ for all samples $p$.
Let $\mathcal{U}_{m,n,k}$ be the set of all $m,n,k$-usage matrices $U$. 
\end{definition}

\begin{figure}[t]
  \center
  \subfloat[\label{fig:F_U_relation}A $4,4$-complete perfect phylogeny tree $T$, here $f_{p,(4,2)} = \sum_{(d,j) \in T_{(4,2)}} u_{p,(d,j)} = u_{p,(4,2)} + u_{p,(2,2)} + u_{p,(2,3)}$ (shaded vertices comprise $T_{(4,2)}$)]{\includegraphics{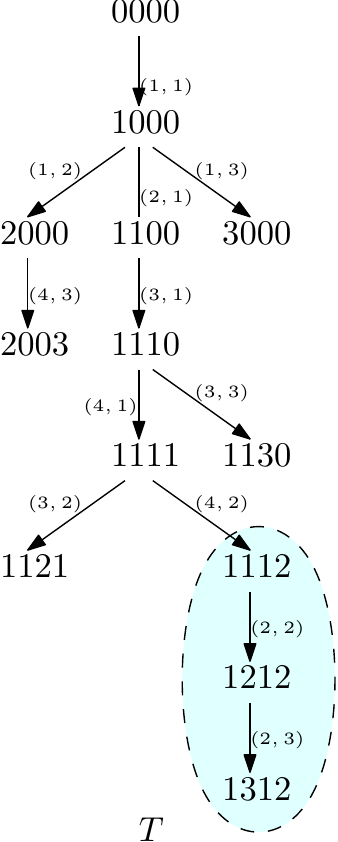}}
  \hspace{.75cm}
  \subfloat[State trees $\mathcal{S}$ determined by $T$]{\includegraphics{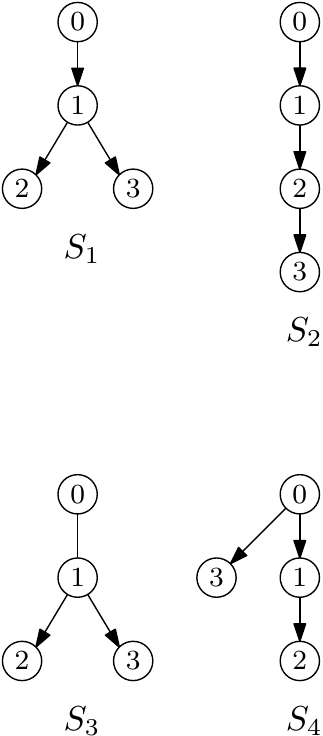}}
  \hspace{.75cm}
  \subfloat[\label{fig:G}$G(A)$, red edges denote a spanning tree rooted at $0000$ that corresponds to $T$]{\includegraphics{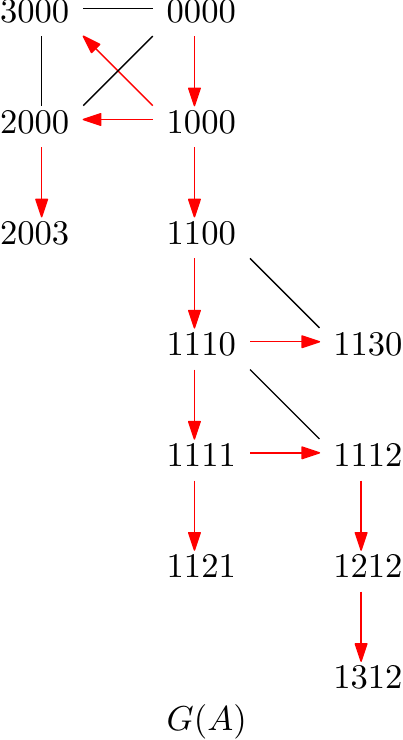}}
  \hspace{.75cm}
  \subfloat[\label{fig:A_example}$4,4$-complete perfect phylogeny matrix $A =
    \theta(T)$. Note that entries in red correspond to the first two conditions
      of Definition~\ref{def:A}]{\includegraphics{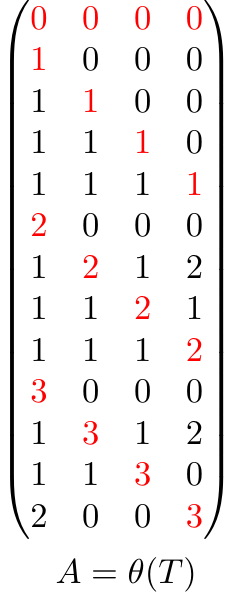}}
  \caption{\textbf{Concepts of the \ProblemFull}}
  \label{fig:concepts}
\end{figure}

Given $\mathcal{F}$, the goal is to infer a perfect phylogeny $T$ and a usage matrix $U$ such that mixing the leaves of $T$ according to $U$ results in $\mathcal{F}$, which was stated as Problem~\ref{prob:1} in the main text.
\begin{prob1}[\ProblemFull~(\Problem)]
  Given an $k \times m \times n$ frequency tensor $\mathcal{F} = [[f_{p,(c,i)}]]$, find a $n,k$-complete perfect phylogeny tree $T$ and a $m,n,k$-usage matrix $U = [u_{p, (c,i)}]$
  such that $f_{p,(c,i)} = \sum_{(d,j) \in T_{(c,i)}} u_{p,(d,j)}$ for all character-state pairs $(c,i)$ and all samples $p$.
\end{prob1}

In the remainder of this section we show how Problem~\ref{prob:1} can be restated as a linear algebra problem.
We start by observing that each vertex $v_{(c,i)}$ of $T$ defines a \emph{state vector} $\mathbf{a}_{(c,i)} \in \{0,\ldots,k-1\}^n$ indicating the state of each character at that vertex.
The root vertex $v_{(*,0)}$ has state vector $\mathbf{a}_{(*,0)} = (0, \ldots, 0)$, i.e.\ $a_{{(*,0)},d} = 0$ for each character $d \in [n]$. The state vector $\mathbf{a}_{(c,i)}$ of the remaining vertices $v_{(c,i)} \neq v_{(*,0)}$ is the same as the state vector $\mathbf{a}_{\pi(c,i)}$ of the parent vertex $v_{\pi(c,i)}$ except at character $c$ where the state is $i$. 
The state vectors of all the vertices of an $n,k$-complete perfect phylogeny tree $T$ correspond to an $(n(k-1) + 1) \times n$ matrix $A$---see Figure~\ref{fig:A_example}.

We now define a subset of matrices $A \in \{0,\ldots,k-1\}^{(n(k-1) + 1) \times n}$ that we call \emph{$n,k$-complete perfect phylogeny matrices} whose rows encode the state vectors of the vertices of an $n,k$-complete perfect phylogeny tree $T$.
Let $A \in \{0,\ldots,k-1\}^{(n(k-1) + 1) \times n}$.
We define $G(A)$ as the undirected graph whose vertices correspond to the rows of $A$, and whose edges set consists of all pairs of vertices whose corresponding state vectors differ at exactly one position, i.e.\ have Hamming distance 1.
We require that $G(A)$ is connected (Figure~\ref{fig:G}), that every row $\mathbf{a}_{(c,i)}$ of $A$ (where $i \in [k]$) introduces the character-state pair $(c,i)$ and that there is a row $\mathbf{a}_{(*,0)}$ that contains only 0-s (Figure~\ref{fig:A_example}). 
Formally, we say that $A$ is an $n,k$-complete perfect phylogeny matrix if the following holds.
\begin{definition}
  \label{def:A}
  Matrix $A = [a_{(c,i),d}] \in \{0,\ldots,k-1\}^{(n(k-1) + 1) \times n}$ is a \emph{$n,k$-complete perfect phylogeny matrix} provided $a_{(*,0),d} = 0$ for all characters $d$, $a_{(c,i),c} = i$ for all character-state pairs $(c,i)$ and $G(A)$ is connected. Let $\mathcal{A}_{n,k}$ be the set of all $n,k$-complete perfect phylogeny matrices.
\end{definition}

Unlike the general multi-state perfect phylogeny problem~\cite{DBLP:conf/icalp/1992}, we can recognize complete perfect phylogeny matrices in polynomial time, as these matrices form a restricted subset of multi-state perfect phylogeny matrices whose rows unambiguously encode all the vertices of a corresponding tree.
We relate complete perfect phylogeny trees to complete perfect phylogeny matrices by defining the following function.

\begin{definition}
  The function $\theta : \mathcal{T}_{n,k} \to \mathcal{A}_{n,k}$ maps a complete perfect phylogeny tree $T \in \mathcal{T}_{n,k}$ to the complete perfect phylogeny matrix $\theta(T) = A = [a_{(c,i),d}]$ where
  \begin{equation}
  a_{(c,i),d} = 
  \begin{cases}
    0, & \mbox{if $(c,i) = (*,0)$,}\\
    i, & \mbox{if $d = c$,}\\
    a_{\pi(c,i),d}, & \mbox{if $d \neq c$.}
  \end{cases}
\end{equation}
\end{definition}

\begin{lemma}
  The function $\theta : \mathcal{T}_{n,k} \to \mathcal{A}_{n,k}$ is a surjection.
\end{lemma}
\begin{proof}
  The set of complete perfect phylogeny trees corresponding to a matrix $A \in \mathcal{A}_{n,k}$ is exactly the set of spanning trees of $G(A)$ rooted at $v_{(*,0)}$. This set is nonempty as by Definition~\ref{def:A}, $G(A)$ is connected and thus has at least one spanning tree for any $A \in \mathcal{A}_{n,k}$. \qed
\end{proof}

We now have the following convenient parameterization of the problem.
We define matrix $A_i = [a^i_{(d,j),c}]$ as
\[
  a^i_{(d,j),c} = 
  \begin{cases}
    1, &\mbox{if $a_{(d,j),c} = i$,}\\
    0, & \mbox{otherwise.}
  \end{cases}
\]
Note that $\sum_{i=0}^{k-1} i A_i = A$.
Since each sample is a mixture of the vertices of $T$, captured by the complete perfect phylogeny matrix~$A$, with proportions defined in the usage matrix $U$, the observed frequency tensor $\mathcal{F} = [F_i]$ satisfies
\begin{equation}
  \label{eq:main_path}
  F_i = U A_i
\end{equation}
for all states $i \in \{0,\ldots,k-1\}$.
Assuming no errors in $\mathcal{F}$, our goal is thus to find $U \in \mathcal{U}_{m,n,k}$ and $A \in \mathcal{A}_{n,k}$ satisfying~\eqref{eq:main_path}.
We thus may restate Problem~\ref{prob:1} as follows.
\begin{problem}[\ProblemFull~(\Problem)]
  Given an $k \times m \times n$ frequency tensor $\mathcal{F} = [F_i] = [[f_{p,(c,i)}]]$, find a $n,k$-complete perfect phylogeny tree $T$ and a $m,n,k$-usage matrix $U = [u_{p, (c,i)}]$
  such that $F_i = U A_i$ for all $i \in \{0,\ldots,k-1\}$ where $A = \theta(T)$.
\end{problem}

\subsection{Uniqueness of $U$ given $\mathcal{F}$ and $T$}

Remarkably, $\mathcal{F} = [F_i]$ and $A \in \mathcal{A}_{n,k}$ \emph{uniquely} define the matrix $U$ such that $F_i = U A_i$ for all states $i$ as we prove in the following.

We start by defining a set of $(n(k-1) + 1) \times nk$ binary matrices $B_{n,k}$ that are in 1-1 correspondence to $A_{n,k}$. 
We do so by defining the undirected graph $H(B)$ for a matrix $B \in \{0,1\}^{(n(k-1) + 1) \times nk}$.
The vertices of $H(B)$ correspond to the rows of $B$ and there is an edge in $H(B)$ if and only if the two corresponding state vectors differ at exactly two positions, i.e.\ have Hamming distance~2.
Formally, we define a \emph{binary $n,k$-complete perfect phylogeny matrix} as follows.
\begin{definition}
  \label{def:path_B}
  A matrix $B = [b_{(d,j),(c,i)}] \in \{0,1\}^{(n(k-1) + 1) \times nk}$ matrix is a \emph{binary $n,k$-complete perfect phylogeny matrix} provided 
  \begin{itemize}
    \item $\sum_{c=1}^n \sum_{i = 0}^{k-1} b_{(d,j),(c,i)} = n$ where $(d,j) \in [n] \times [k-1]$, 
    \item $\sum_{c=1}^n b_{(d,j),(c,i)} = 1$ where $(d,j) \in [n] \times [k-1]$ and $i \in \{0,\ldots,k-1\}$, 
    \item $b_{(*,0),(c,0)} = 1$ where $c \in [n]$,
    \item $b_{(c,i),(c,i)} = 1$ for all $(c,i) \in [k] \times [n-1]$ and 
    \item $H(B)$ is connected. 
  \end{itemize}
  Let $\mathcal{B}_{n,k}$ be the set of all binary $n,k$-complete perfect phylogeny matrices.
\end{definition}

We now define the following function $\psi$ that maps a complete perfect phylogeny matrix $A$ to a binary matrix.
\begin{definition}
  The function $\psi$ maps a complete perfect phylogeny matrix $A \in \mathcal{A}_{n,k}$ to the binary matrix $\psi(A) = B = [A_0 \ldots A_{k-1}]$.
\end{definition}

We now show that $\psi(A)$ is a binary $n,k$-complete perfect perfect phylogeny matrix for each $A \in \mathcal{A}_{n,k}$, and that $\psi$ is in fact a bijection.

\begin{lemma}
  \label{lem:A_equivalent_B}
  The function $\psi$ is a bijection between $\mathcal{A}_{n,k} \rightarrow \mathcal{B}_{n,k}$.
\end{lemma}
\begin{proof}
  Let $A \in \mathcal{A}_{n,k}$. We claim that $B = \psi(A) = [A_0 \ldots A_{k-1}] \in \mathcal{B}_{n,k}$. Recall that $A_i = [a^i_{(d,j),c}]$ where 
\[
  a^i_{(d,j),c} = 
  \begin{cases}
    1, &\mbox{if $a_{(d,j),c} = i$,}\\
    0, & \mbox{otherwise.}
  \end{cases}
\]
  Therefore $a^i_{(d,j),c} = b_{(d,j),(c,i)}$. 
  We thus have that $\sum_{c=1}^n \sum_{i = 0}^{k-1} b_{(d,j),(c,i)} = n$ and $\sum_{c=1}^n b_{(d,j),(c,i)} = 1$ where $i \in \{0,\ldots,k-1\}$. 
  Moreover, because $a_{(*,0),d} = 0$ for all $d \in [n]$, we have that $b_{(*,0),(d,0)} = 1$. 
  Also, as $a_{(c,i),c} = i$, we have $b_{(c,i),(c,i)} = 1$. 
  Furthermore, $G(A)$ and $H(B)$ are isomorphic with $u_{(c,i)} \leftrightarrow v_{(c,i)}$ where $u_{(c,i)} \in V(G)$, $v_{(c,i)} \in V(H)$. 
  Thus, $H(B)$ is connected as $G(A)$ is connected. 
  Hence, $A \in \mathcal{A}_{n,k}$.

  Let $B = [b_{(d,j),(c,i)}] \in \mathcal{B}_{n,k}$. 
  We claim that $A = [a_{(d,j),c}] \in \mathcal{A}_{n,k}$ where $a_{(d,j),c} = i$ such that $b_{(d,j),(c,i)} = 1$. 
  Since $b_{(*,0),(c,0)} = 1$ where $c \in [n]$, we have that $a_{(*,0),c} = 0$. 
  Moreover, as $b_{(c,i),(c,i)} = 1$ for all $(c,i) \in [k] \times [n-1]$, we have that $a_{(c,i),c} = i$. 
  Furthermore, $H(B)$ and $G(A)$ are isomorphic with $u_{(c,i)} \leftrightarrow v_{(c,i)}$ for $u_{(c,i)} \in V(G)$, $v_{(c,i)} \in V(H)$. 
  Thus, $G(A)$ is connected as $H(B)$ is connected. 
  Hence, $B \in \mathcal{B}_{n,k}$.
  \qed
\end{proof}

We now flatten frequency tensor $\mathcal{F} = [F_i]$ into the $m \times nk$ matrix $F = [F_0 \ldots F_{k-1}]$ and prove that the problem is equivalent to factorizing $F$.

\begin{lemma}
  Let $\mathcal{F} = [F_i]$ be frequency tensor and let $U \in \mathcal{U}_{m,n,k}$ be a usage matrix. There exists a matrix $A \in \mathcal{A}_{n,k}$ such that $F_i = U A_i$ for all states $i \in \{0,\ldots,k-1\}$ if and only if there exists a matrix $B \in \mathcal{B}_{n,k}$ such that 
  \begin{equation}
    \label{eq:F=UB}
    F = [F_0 \ldots F_{k-1}] = UB.
  \end{equation}
\end{lemma}
\begin{proof}
  By Lemma~\ref{lem:A_equivalent_B}, let $A \in \mathcal{A}_{n,k}$ and $B \in \mathcal{B}_{n,k}$ be corresponding matrices. Note that $B = \psi(A) = [A_0 \ldots A_{k-1}]$ where $A_i = [a^i_{(d,j),c}]$ such that 
  \[
  a^i_{(d,j),c} = 
  \begin{cases}
    1, &\mbox{if $a_{(d,j),c} = i$,}\\
    0, & \mbox{otherwise.}
  \end{cases}
  \]
  Since $a^i_{(d,j),c} = b_{(d,j),(c,i)}$, we have
  \begin{align*}
    f_{p,(c,i)} & = u_{p,(*,0)} a^i_{(*,0),c)} + \sum_{d=1}^n \sum_{j=1}^{k-1} u_{p,(d,j)} a^i_{(d,j),c} = u_{p,(*,0)} b_{(*,0),(c,i)} + \sum_{d=1}^n \sum_{j=1}^{k-1} u_{p,(d,j)} b_{(d,j),(c,i)}
  \end{align*}
  for all $p \in [m]$, $c \in [n]$ and $i \in \{0,\ldots,k-1\}$.
  \qed
\end{proof}

Hence, we may restate Problem~\ref{prob:1} as a matrix factorization problem.
\begin{problem}[\ProblemFull~(\Problem)]
  Given an $k \times m \times n$ frequency tensor $\mathcal{F} = [F_i] = [[f_{p,(c,i)}]]$, find a $n,k$-complete perfect phylogeny tree $T$ and a $m,n,k$-usage matrix $U = [u_{p, (c,i)}]$
  such that $F = U B$ where $F = [F_0 \ldots F_{k-1}]$ and $B = \psi(\theta(T))$.
\end{problem}

We now have all the ingredients to show that $\mathcal{F}$ and $T$ uniquely define a matrix $U$.
We first show that any matrix $B \in \mathcal{B}_{n,k}$ has full row rank.

\begin{lemma}
  \label{lem:full_row_rank}
  Any matrix $B \in \mathcal{B}_{n,k}$ has row rank $n(k-1) + 1$.
\end{lemma}
\begin{proof}
  By Definition~\ref{def:path_B}, we have that $B = \begin{pmatrix}\mathbf{1} & \mathbf{0}\\C & D\end{pmatrix}$ where $C$ has dimensions $n(k-1) \times n$, $D$ has dimensions $n(k-1) \times n(k-1)$, $\mathbf{1}$ is the $1 \times n$ matrix of all 1-s and $\mathbf{0}$ is the $1 \times n(k-1)$ matrix of all 0-s.
  We show that the square submatrix $D = [b_{(d,j),(c,i)}]$, where $(c,i), (d,j) \in [n] \times [k-1]$, has full rank by performing Gaussian elimination according to a breadth-first search on $H(B)$, starting from the all-zero ancestor $v_{(*,0)}$. 
  Let $\ell(v)$ denote the breadth-first search (BFS) level of vertex $v \in V(H(B))$. 
  Note that $\ell(v_{(*,0)}) = 0$.
  We claim that this process results in the $n(k-1) \times n(k-1)$ identity matrix $I$ where row $\mathbf{i}_{(c,i)}$ corresponds to vertex $v_{(c,i)} \in V(H(B)) \setminus \{v_{(*,0)}\}$.

  We show this constructively by induction on the BFS level $l$. 
  The claim is that at BFS level $l$ all rows $\mathbf{i}_{(c,i)}$ where $\ell(v_{(c,i)}) \leq l$ have been generated from $D$ using elementary row operations.
  Initially, at $l=1$, for each vertex $v_{(d,j)}$ with BFS level $\ell(v_{(d,j)}) = l = 1$ it holds that $b_{(d,j),(d,j)} = 1$ and $b_{(d,j),(c,i)} = 0$ for all $(c,i) \in [n] \times [k-1] \setminus \{(d,j)\}$. 
  Therefore the vertices $v_{(d,j)}$ at BFS level 1 correspond directly to rows $\mathbf{i}_{(d,j)}$ of $I$.
  At every iteration $l>1$, we generate, using elementary row operations, the rows $\mathbf{i}_{(d,j)}$ of $I$ that correspond to vertices $v_{(d,j)}$ at BFS level $\ell(v_{(d,j)}) = l$.
  In order to obtain row $\mathbf{i}_{(d,j)}$ of $I$, we subtract from $\mathbf{b}_{(d,j)}$ the rows $\mathbf{i}_{(c,i)}$ where $b_{(d,j),(c,i)} = 1$ and $(c,i) \neq (d,j)$. 
  Observe that by Definition~\ref{def:path_B}, $H(B)$ is connected and that every character-state pair $(c,i)$ must have been introduced by vertex $v_{(c,i)}$, which must on a path to the root $v_{(*,0)}$ from $v_{(d,j)}$.
  Hence, $\ell(v_{(c,i)}) < l$ for every $(c,i) \in [n] \times [k-1] \setminus \{(d,j)\}$ where $b_{(d,j),(c,i)} = 1$.
  Since the corresponding vertices $v_{(c,i)}$ are at a BFS level strictly smaller than $l$, the corresponding rows $\mathbf{i}_{(c,i)}$ have already been generated by the induction hypothesis. 
  Therefore at the final iteration, we obtain the identity matrix $I$ from $D$ using elementary row operations.
  It thus follows that $D$ is full rank.

  Since $D$ is full rank, the row rank of $\begin{pmatrix}C & D\end{pmatrix}$ equals the rank of $D$, which is $n(k-1)$. 
  Furthermore, the first row $\begin{pmatrix} \mathbf{1} & \mathbf{0}\end{pmatrix}$ of $B$ cannot be expressed as a linear combination of $\begin{pmatrix} C & D\end{pmatrix}$. 
  This implies that the row rank of $B$ is $n(k-1) + 1$. \qed
\end{proof}

This means that given the $m \times nk$ frequency matrix $F$ and the $(n(k-1) + 1) \times nk$ binary complete perfect phylogeny matrix $B$, there exists a \emph{unique} $m \times (n(k-1) + 1)$ matrix $U$ such that \eqref{eq:F=UB} holds, i.e.\ $U = FB^{-1}$ where the $nk \times (n(k-1) + 1)$ matrix $B^{-1}$ is the right inverse of $B$ such that $B B^{-1} = I$ where $I$ is the $(n(k-1) + 1) \times (n(k-1)+1)$ identity matrix.
Given $F$ and $T$, we now define the unique matrix $U = [u_{p,(c,i)}]$ without explicitly computing the right inverse of $B$. 
We do this using the notion of a \emph{descendant set} of character-state pairs $(c,i)$. 
The descendant set $D_{(c,i)}$ is the set of states for character $c$ that are descendants of $v_{(c,i)}$ in $T$.
Formally, $D_{(c,i)} = \{ j \mid (c,i) \prec_T (c,j) \}$.
We denote by $T_{(c,i)}$ the subtree of $T$ consisting of all vertices that have state $i$ for character $c$, and by $\overline{T}_{(c,i)}$ the subtree rooted at vertex $v_{(c,i)}$.
The descendant set of a character precisely determines the relationship between $\overline{T}_{(c,i)}$ and $T_{(c,i)}$; namely we have $\overline{T}_{(c,i)} = \bigcup_{l \in D_{(c,i)}} T_{(c,l)}$. 
In the two-state $(k=2)$ case, we have that $T_{(c,1)} = \overline{T}_{(c,1)}$ (see Figure~\ref{fig:t_vs_tbar}).
Recall that $f_{p,(c,i)} = \sum_{(d,j) \in T_{(c,i)}} u_{p,(d,j)}$.
We define the \emph{cumulative frequency} $f^+_p(D_{(c,i)}) = \sum_{l \in D_{(c,i)}} f_{p,(c,l)}$.
In the following lemma we show that given $T \in \mathcal{T}_{n,k}$, the cumulative frequencies $f^+_p(D_{(c,i)})$ for the descendant sets defined by $T$ \emph{uniquely} determine the usage matrix $U$.
For an intuitive explanation of the usage equation \eqref{eq:U-app}, see Figure~\ref{fig:usage_equation}.


\begin{figure}[t]
  \centering
  \includegraphics{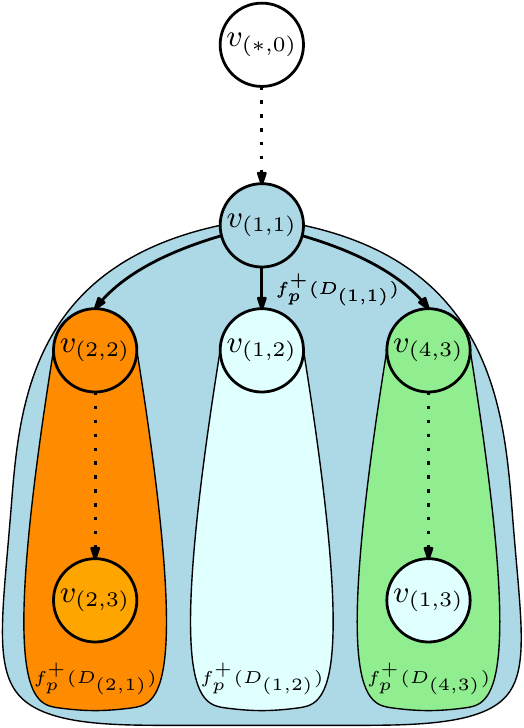}
  \caption{\textbf{Example of the usage equation \eqref{eq:U-app}.} Recall that $f_{p,(c,i)}$ is the sum of the usages of the vertices that have state $i$ for character $c$, i.e.\ $f_{p,(c,i)} = \sum_{(d,j) \in T_{(c,i)}} u_{p,(d,j)}$. The observation here is that the cumulative frequency $f^+_p(D_{(c,i)})$ equals the sum of the usages of all vertices in the subtree of $T$ rooted at vertex $v_{(c,i)}$, i.e.\ $f^+_p(D_{(c,i)}) = \sum_{l \in D_{(c,i)}} f_{p,(c,l)} = \sum_{l \in D_{(c,i)}} \sum_{(d,j) \in T_{(c,l)}} u_{p,(d,j)}$. Hence, $u_{p,(1,1)} = f^+_p(D_{(1,1)}) - (f^+_p(D_{(2,2)}) + f^+_p(D_{(1,2)}) + f^+_p(D_{(4,3)})) = (f_{p,(1,1)} + f_{p,(1,2)} + f_{p,(1,3)}) - (f_{p,(2,2)} + f_{p,(2,3)} + f_{p,(1,2)} + f_{p,(4,3)})$.}
  \label{fig:usage_equation}
\end{figure}

\begin{lemU}
  Let $T \in \mathcal{T}_{n,k}$ and $\mathcal{F} = [[f_{p,(c,i)}]]$ be a frequency tensor.
  For a character-state pair $(c,i)$ and sample $p$, let
  \begin{equation}
    u_{p,(c,i)} = f^+_p(D_{(c,i)}) - \sum_{(d,j) \in \delta(c,i)} f^+_p(D_{(d,j)}).
    \label{eq:U-app}
  \end{equation}
  Then $U = [u_{p,(c,i)}]$ is the unique matrix whose entries satisfy $f_{p,(c,i)} = \sum_{(d,j) \in T_{(c,i)}} u_{p,(d,j)}$.
\end{lemU}
\begin{proof}
  Let $(c,i)$ be a character-state pair and $p$ be a sample.
  Let $A = \theta(T)$ and $B = \psi(A)$.
  Recall that $T_{(c,i)}$ is the set of vertices $\{v_{(d,j)} \mid b_{(d,j),(c,i)} = 1 \}$.
  Note that by definition, the vertices of $T_{(c,i)}$ induce a connected subtree in $T$.
  We thus need to show that
  \[
    f_{p,(c,i)} = u_{p,(*,0)} b_{(d,j),(c,i)} + \sum_{d=1}^n \sum_{j=1}^{k-1} u_{p,(d,j)} b_{(d,j),(c,i)} = \sum_{(d,j) \in T_{(c,i)}} u_{p,(d,j)}.
  \]
  We introduce the following shorthand $\Delta(c,i) = \bigcup_{(d,j) \in T_{(c,i)}} \delta(d,j) \setminus T_{(c,i)}$, which is the set of vertices $\{v_{(d,j)}\}$ that are not in $T_{(c,i)}$ but whose parent $v_{\pi(d,j)}$ is in $T_{(c,i)}$.
  Thus,
  \begin{align*}
    f_{p,(c,i)} & = \sum_{(d,j) \in T_{(c,i)}} u_{p,(d,j)}\\
    & = \sum_{(d,j) \in T_{(c,i)}}\left( f^+_p(D_{(d,j)}) - \sum_{(e,l) \in \delta(d,j)} f^+_p(D_{(e,l)}) \right).
  \end{align*}
  Observe that in the equation above, for every $(d,j) \in T_{(c,i)} \setminus \{(c,i)\}$, there are two terms $f^+_p(D_{(d,j)})$ and $-f^+_p(D_{(d,j)})$, which consequently cancel out.
  The remaining terms are $f^+_p(D_{(c,i)})$, and $-f_p(D_{(c,l)})$ for each $(c,l) \in \Delta(c,i)$.
  In addition, we have that the state trees $\{ D_{(c,l)} \}$ corresponding to the elements of $(c,l)  \in \Delta(c,i)$ are pairwise disjoint.
  Moreover, $D_{(c,i)} \setminus \{i\} = \bigcup_{(c,l) \in \Delta(c,i)} D_{(c,l)}$.
  Thus,
  \begin{align*}
    f_{p,(c,i)} & = f^+_p(D_{(c,i)}) - \sum_{(c,l) \in \Delta(c,i)} f^+_p(D_{(c,l)})\\
    & = f_{p,(c,i)} + \sum_{(c,l) \in \Delta(c,i)} f^+_p(D_{(c,l)}) - \sum_{(c,l) \in \Delta(c,i)} f^+_p(D_{(c,l)})\\
    & = f_{p,(c,i)}. 
  \end{align*}
  By Lemma~\ref{lem:full_row_rank}, the equation $F = UB$ has only one solution given $F$ and $B$. Thus $U = [u_{p,(c,i)}]$ is the unique matrix such that $f_{p,(c,i)} = \sum_{(d,j) \in T_{(c,i)}} u_{p,(d,j)}$ for all samples $p$ and character-state pairs $(c,i)$.
  \qed
\end{proof}

\subsection{Combinatorial Characterization of the \Problem}

We say that $T$ \emph{generates} $\mathcal{F}$ if the corresponding matrix $U$, defined by \eqref{eq:U}, is a usage matrix.
It turns out that positivity of the values $u_{p,(c,i)}$ is a necessary and sufficient condition for $T$ to generate $\mathcal{F}$.
We show this in the following theorem.

\begin{thm1}
  A complete perfect phylogeny tree $T$ generates $\mathcal{F}$ if and only if 
  \begin{equation}
    \tag{\ref{eq:MSSC}} f^+_p(D_{(c,i)}) - \sum_{(d,j) \in \delta(c,i)} f^+_p(D_{(d,j)}) \ge 0
  \end{equation}
  for all character-state pairs $(c,i)$ and samples $p$.
\end{thm1}
\begin{proof}
  $(\Rightarrow)$ We start by proving the forward direction. 
  Let $F = [F_0 \ldots F_{k-1}]$, $T$ be a tree that generates $F$, $B = \psi(\theta(T))$ be the corresponding binary matrix of $T$ and let $U = [u_{p,(c,i)}]$ be the corresponding usage matrix.
  Since $T$ generates $F$, we have that $u_{p,(c,i)} \geq 0$ for all character-state pairs $(c,i)$ and samples $p$.
  By Lemma~\ref{lem:U}, \eqref{eq:MSSC} thus holds.

  $(\Leftarrow)$ As for the reverse direction, we need to show that if \eqref{eq:MSSC} is met, the matrix $U$ as defined in Lemma~\ref{lem:U} is a usage matrix.
  Let $p \in [m]$. 
  We prove this direction by showing that (i) $u_{p,(c,i)} \geq 0$ for all character-state pairs $(c,i)$, and (ii) $u_{p,(*,0)} + \sum_{c=1}^n \sum_{i=1}^{k-1} u_{p,(c,i)} = 1$.
  \begin{enumerate}
    \item[(i)] By Lemma~\ref{lem:U} and \eqref{eq:MSSC}, we have that $u_{p,(c,i)} \geq 0$.
    \item[(ii)] 
      We now have
      \[
        u_{p,(*,0)} + \sum_{c=1}^n\sum_{i=1}^{k-1} u_{p,(c,i)} = f^+_p(D_{(*,0)}) - \sum_{(d,j) \in \delta(*,0)} f^+_p(D_{(d,j)}) + \sum_{c=1}^n \sum_{i=1}^{k-1} \left( f^+_p(D_{(c,i)}) - \sum_{(d,j) \in \delta(c,i)} f^+_p(D_{(d,j)})\right).
      \]
      Observe that for each $(c,i) \in [n] \times [k-1]$ there are exactly two terms in the above equation: $+f^+_p(D_{(c,i)})$ when $v_{(c,i)}$ is considered as a parent and $-f^+_p(D_{(c,i)})$ when $v_{(c,i)}$ is considered as a child. 
      Hence, all these terms cancel out.
      Since $D_{(*,0)} = \{0,\ldots,k-1\}$ and $\sum_{i=0}^{k-1} f_{p,(c,i)} = 1$, we have that $f^+_p(D_{(*,0)}) = 1$.
      Thus, $u_{p,(*,0)} + \sum_{c=1}^n\sum_{i=1}^{k-1} u_{p,(c,i)} = 1$.
  \end{enumerate}
  \qed
\end{proof}

Although the sets $D_{(c,i)}$ are \emph{a priori} unknown, we will show that the following condition \eqref{eq:MSAC} must hold for any tree $T$ that generates $\mathcal{F}$ where $(c,i) \prec_T (d,j)$.

\begin{def1}
  Let $(c,i)$ and $(d,j)$ be distinct character-state pairs and let $D_{(c,i)}, D_{(d,j)} \subseteq \{0,\ldots,k-1\}$.
  A pair $(D_{(c,i)}, D_{(d,j)})$ is a \emph{valid descendant set pair} provided 
  \begin{equation}
    \tag{\ref{eq:MSAC}}
    f^+_p(D_{(c,i)}) - f^+_p(D_{(d,j)}) \ge 0
  \end{equation}
  for all samples $p$; and additionally if $c = d$ then $D_{(c,j)} \subsetneq D_{(c,i)}$.
\end{def1}

It turns out that there are potentially many valid descendant set pairs as shown by the following lemma.
\begin{lem2}
  Let $(D_{(c,i)}, D_{(d,j)})$ be a valid descendant set pair. If $D_{(c,i)} \subseteq D'_{(c,i)}$ and $D'_{(d,j)} \subseteq D_{(d,j)}$ then $(D'_{(c,i)}, D'_{(d,j)})$ is a valid descendant set pair.
\end{lem2}
\begin{proof}
  Let $D_{(c,i)} \subseteq D'_{(c,i)}$ and $D'_{(d,j)} \subseteq D_{(d,j)}$.
  Let $p \in [m]$.
  Since $f_{p,(c,i)} \geq 0$ (by Definition~\ref{def:F}), $D_{(c,i)} \subseteq D'_{(c,i)}$ and $D'_{(d,j)} \subseteq D_{(d,j)}$, we have that $f^+_p(D'_{(c,i)}) \ge f^+_p(D_{(c,i)})$ and $f^+_p(D_{(d,j)} \ge f^+_p(D'_{(d,l)})$.
Moreover, if $c = d$ then $D_{(d,j)} \subsetneq D_{(c,i)}$ and thus $D'_{(d,j)} \subsetneq D'_{(c,i)}$.
  Hence, $(D'_{(c,i)}, D'_{(d,j)})$ is a valid descendant set pair. \qed
\end{proof}

Since $v_{(*,0)}$ is the all-zero ancestor, we have the following corollary that describes the \emph{extreme} valid descendant set pair.

\begin{corollary}
  \label{cor:supersub}
    Let $T$ be a complete perfect phylogeny tree that generates $\mathcal{F}$. If $(c,i) \prec_T (d,j)$ and $(c,i) \neq (c,0)$ then $D_{(c,i)} = [k-1]$ and $D_{(d,j)} = \{j\}$ form a valid descendant set pair $(D_{(c,i)}, D_{(d,j)})$.
\end{corollary}

The following proposition shows that \eqref{eq:MSAC} is a necessary condition to solutions of the \Problem.

\begin{prop1}
  Let $T$ be a complete perfect phylogeny tree that generates $\mathcal{F}$. If $(c,i) \prec_T (d,j)$ then there exist a valid descendant set pair $(D_{(c,i)}, D_{(d,j)})$.
\end{prop1}
\begin{proof}
  Let $(c,i) \prec_T (d,j)$.
  Let $P = v_{(c_1,i_1)}, \ldots, v_{(c_t,i_t)}$ where $v_{(c_1,i_1)} = v_{(c,i)}$ and $v_{(c_t,i_t)} = v_{(d,j)}$ be the unique path from $v_{(c,i)}$ to $v_{(d,j)}$ in $T$.
  Assume for a contradiction that there exists no valid descendant set pair $(D_{(c,i)}, D_{(d,j)})$.
  By Corollary~\ref{cor:supersub}, we thus have that $D_{(c,i)} = [k-1]$ and $D_{(d,j)} = \{j\}$ do not form a valid descendant set pair $(D_{(c,i)}, D_{(d,j)})$.
  Now, if $c=d$, we would have a contradiction as $\{j\} \subseteq [k-1]$.
  Therefore, $c \neq d$ and $f^+_p(D_{(c,i)}) - f^+_p(D_{(d,j)}) < 0$.
  By Theorem~\ref{thm:characterization}, we have that $f^+_p(D_{(c_l,i_l)}) - f^+_p(D_{(c_{l+1},i_{l+1})}) \geq 0$ for all $1 \leq l < t$.
  Therefore, $f^+_p(D_{(c_1,i_1)}) - f^+_p(D_{(c_{t},i_{t})}) \ge 0$, which leads to a contradiction.
  \qed
\end{proof}

We now define the multi-state ancestry graph $G_\mathcal{F}$ whose vertices correspond to character-state pairs and whose edges correspond to valid descendant state pairs.
\begin{def1}
The \emph{multi-state ancestry graph} $G_{\mathcal{F}}$ of the frequency tensor $\mathcal{F}$ is an edge-labeled, directed multi-graph $G_{\mathcal{F}} = (V,E)$ whose vertices $v_{(c,i)}$ correspond to character-state pairs $(c,i)$ and whose multi-edges are $(v_{(c,i)},v_{(d,j)})$ for all valid descendant set pairs $(D_{(c,i)}, D_{(d,j)})$.
\end{def1}

Note that $v_{(1,0)},\ldots,v_{(n,0)}$ all refer to the same root vertex $v_{(*,0)}$.
See Figure~\ref{fig:example} for an example multi-state ancestry graph.
We use the labels of the multi-edges to restrict the set of allowed spanning trees by defining a threading as follows.

\begin{def1}
  A rooted subtree $T$ of $G_{\mathcal{F}} = (V,E)$ is a \emph{threaded tree} provided (1) for every pair of adjacent edges $(v_{(c,i)}, v_{(d,j)}), (v_{(d,j)}, v_{(e,l)}) \in E(T)$ with corresponding labels $(D_{(c,i)}, D_{(d,j)})$ and $(D'_{(d,j)}, D'_{(e,l)})$, it holds that $D_{(d,j)} = D'_{(d,j)}$, and (2) 
  for every pair of vertices $v_{(c,i)},v_{(c,j)} \in V(T)$ it holds that $D_{(c,j)} \subseteq D_{(c,i)}$ if and only if $(c,i) \prec_T (c,j)$.
\end{def1}

We now prove that solutions of an \Problem~instance $\mathcal{F}$ correspond to threaded spanning trees of the ancestry graph $G_\mathcal{F}$.

\begin{thm2}
  \label{thm:combinatorial_characterization2}
  A complete perfect phylogeny tree $T$ generates $\mathcal{F}$ if and only if $T$ is a threaded spanning tree of the ancestry graph $G_\mathcal{F}$ such that \eqref{eq:MSSC} holds.
\end{thm2}
\begin{proof}
  $(\Rightarrow)$ Let $T$ be a complete perfect phylogeny tree generating $\mathcal{F}$.
  We claim that $T$ is a threaded spanning tree of the ancestry graph $G_\mathcal{F}$.
  We start by showing that every edge $(v_{(c,i)},v_{(d,j)}) \in E(T)$ is an edge of $G_\mathcal{F}$ labeled by $(D_{(c,i)},D_{(d,j)})$.
  By Theorem~\ref{thm:characterization}, we have that
  $f^+_p(D_{(c,i)}) - \sum_{(d,j) \in \delta(c,i)} f^+_p(D_{(d,j)}) \ge 0$
  for all character-state pairs $(c,i)$ and samples $p$.
  Let $(c,i)$ be a character-state pair.
  By definition, we have that $D_{(c, j)} \subsetneq D_{(c,i)}$ for all $(c,j) \in \delta(c,i)$.
  Moreover, $i \in D_{(c,i)}$ and $j \in D_{(d,j)}$ for all character-state pairs $(d,j) \in \delta(c,i)$.
  Thus, $(D_{(c,i)}, D_{(d,j)})$ is a valid descendant set pair for all character-state pairs $(d,j) \in \delta(c,i)$.
  Hence, every edge $(v_{(c,i)},v_{(d,j)})$ of $T$ is an edge of $G_\mathcal{F}$ labeled by the valid descendant set pair $(D_{(c,i)}, D_{(d,j)})$.
  Thus, $T$ is a tree of $G$.
  Next, we show that $T$ is a \emph{threaded} spanning tree.
  \begin{enumerate}
    \item By definition of $D$, we have that every pair of adjacent edges $(v_{(c,i)}, v_{(d,j)}), (v_{(d,j)}, v_{(e,l)}) \in E(T)$ is labeled by $(D_{(c,i)}, D_{(d,j)})$ and $(D_{(d,j)}, D_{(e,l)})$, respectively.
    \item By definition of $D$, we have that for every edge $(v_{(c,i)}, v_{(d,j)}) \in E(T)$ labeled by $(D_{(c,i)}, D_{(d,j)})$, it holds that $D_{(c,i)} = \{ l \mid (c,i) \prec_T (c,l) \}$ and $D_{(d,j)} = \{ l \mid (d,j) \prec_T (d,l) \}$. 
      Hence, $(c,i) \prec_T (c,j)$ if and only if $D_{(c,j)} \subseteq D_{(c,i)}$.
  \end{enumerate}
  The conditions of Definition~\ref{def:threaded} are thus met. Therefore, $T$ is a threaded spanning tree of $G_\mathcal{F}$.

  $(\Leftarrow)$ Let $T$ be a threaded spanning tree of the ancestry graph $G_\mathcal{F}$ such that $f^+_p(D_{(c,i)}) - \sum_{(d,j) \in \delta(c,i)} f^+_p(D_{(d,j)}) \ge 0$ for all character-state pairs $(c,i)$ and samples $p$.
  Observe that $T$ is a complete perfect phylogeny tree.
  By condition (1) of Definition~\ref{def:threaded}, we have that for all character-state pairs $(c,0) \neq (*,0)$ and $(d,j) \in \delta(c,i)$, adjacent edges $(v_{\pi(c,i)}, v_{(c,i)})$, $(v_{(c,i)}, v_{(d,j)})$ are labeled by $(D_{\pi(c,i)},D_{(c,i)})$ and $(D_{(c,i)},D_{(d,j)})$---where $\pi(c,i)$ is the parent character-state pair of $(c,i)$.
  Hence, we may use $D_{(c,i)}$ to unambiguously denote the descendant state set of the character-state pair $(c,i)$.
  Moreover, by condition (2) of Definition~\ref{def:threaded} and the fact that $T$ is a spanning tree of $G_\mathcal{F}$, we have that $D$ is defined in the same way as in Theorem~\ref{thm:characterization}.
  By Theorem~\ref{thm:characterization}, we thus have that $T$ generates $\mathcal{F}$. \qed
\end{proof}

Figure~\ref{fig:example} shows an example instance $\mathcal{F}$, where $k=3$, $m=2$ and $n=2$, including the multi-state ancestry graph $G_\mathcal{F}$ and a solution $T$ that generates $\mathcal{F}$.

\begin{figure}[t]
  \centering
  \subfloat[Frequency tensor {$\mathcal{F} = [F_0 \, F_1 \, F_2]$ where $m=2$, $n=2$ and $k=3$}]{\includegraphics{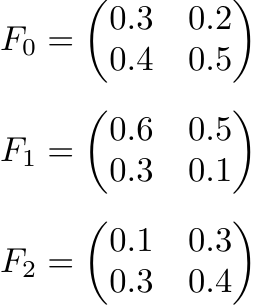}}
  \hspace{.5cm}
  \subfloat[Multi-state ancestry graph $G_\mathcal{F}$. Red edges denote an $n,k$-complete perfect phylogeny tree $T$ that generates $\mathcal{F}$]{\includegraphics[width=.8\textwidth]{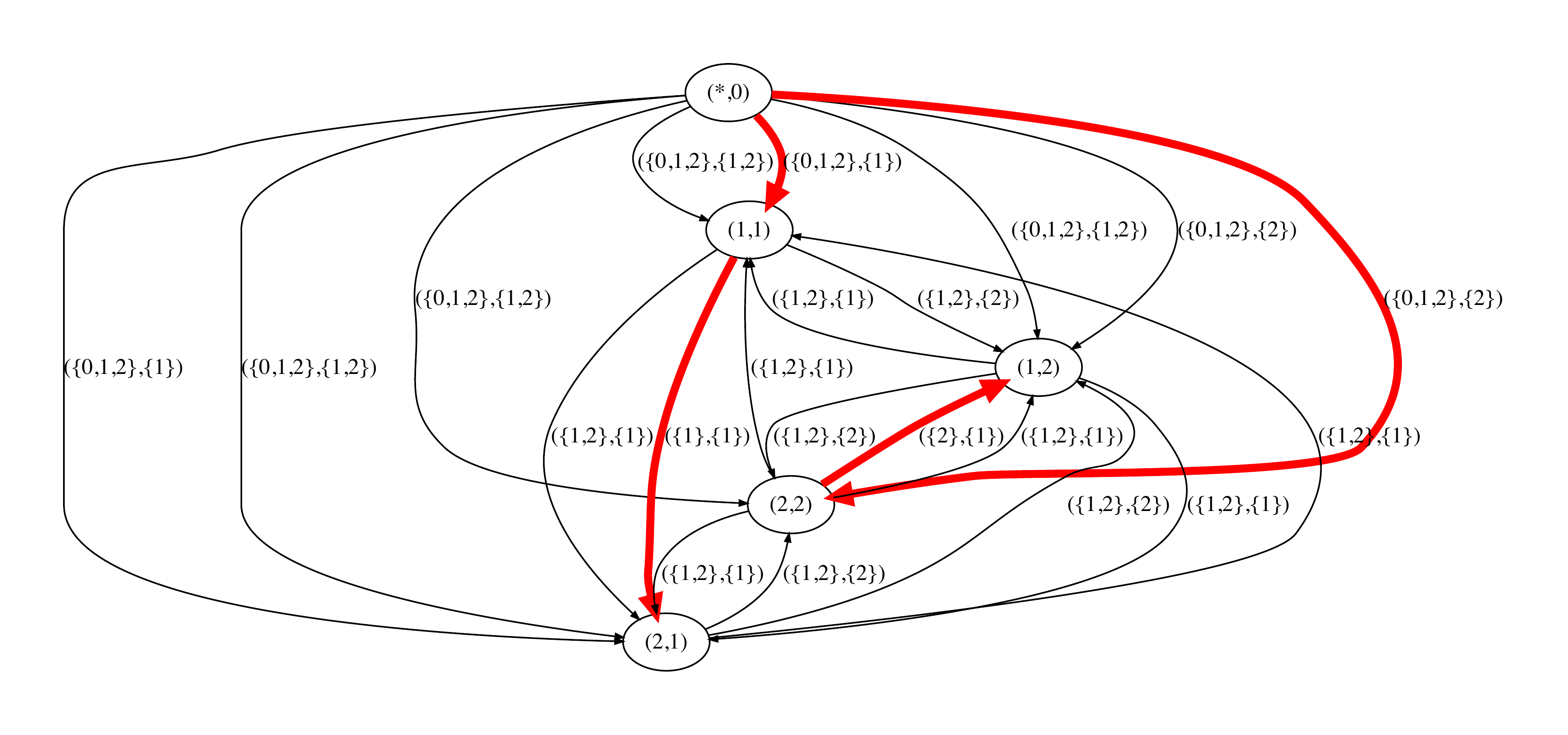}}
  \caption{\textbf{Example \Problem~instance and solution}}
  \label{fig:example}
\end{figure}

\subsection{Problem Complexity}
\label{sec:complexity}

In previous work, we have shown that the problem is NP-complete for general $m$~\cite{ElKebir:2015by}.
An open question was the hardness for constant $m$, which we resolve with the
following lemma.

\begin{thm3}
  The \Problem~is NP-complete even for $m=2$ and $k=2$.
\end{thm3}
\begin{proof}
  Clearly, the problem is in NP---given matrices $U$ and $A$ we can check in
  polynomial time whether $F_i = U A_i$ for all $i \in \{0,\ldots,k-1\}$. We show NP-hardness by reduction
  from SUBSET SUM, which, given nonnegative integers
  $B = \{b_1, \ldots, b_t\}$ and $d$, asks whether there
  exists a subset $B' \subseteq B$ whose sum equals $d$. This problem is
  NP-complete~\cite{Garey:1979}. 
  
  Let $e = \sum_{\ell=1}^t c_\ell$.
  Without loss of generality assume that $e > 0$, $b_\ell < d$ for all $\ell \in [t]$
  and that $b_\ell \leq b_{\ell+1}$ for all $\ell \in [t-1]$. The corresponding frequency
  tensor $\mathcal{F} = [F_0 \, F_1]$ is then defined as follows.
  \begin{align*}
    F_0 & = \frac{1}{e}
    \begin{pmatrix}
      e - d & d & e - b_1 & e - b_2 & \ldots & e - b_t\\
      d & e - d & e - t\epsilon & e - (t-1)\epsilon & \ldots & e - \epsilon
    \end{pmatrix}\\
    F_1 & = \frac{1}{e}
    \begin{pmatrix}
      d & e - d & b_1 & b_2 & \ldots & b_t\\
      e - d & d & t \epsilon & (t-1) \epsilon & \ldots & \epsilon
    \end{pmatrix}.
  \end{align*}
  Note that $\mathcal{F}$ is a $k \times m \times n$ tensor where $k = 2$, $m = 2$ and $n = t + 2$.
  Also note that the normalization factor $\frac{1}{e}$ ensures that each $f_{p,(c,i)} \in [0,1]$ and that $f_{p,(c,0)} + f_{p,(c,1)} = 1$ for all $p \in [m]$ and $c \in [n]$.
  Clearly $\mathcal{F}$ can be obtained in polynomial time from a SUBSET SUM instance. By construction, the ancestry
  graph $G_\mathcal{F}$ consists of a root vertex $v_{(*,0)}$ with outgoing edges to all other vertices $\{d, e-d, b_1, \ldots, b_t\}$. In
  addition, these vertices induce a complete bipartite graph with vertex sets
  $\{d,e-d\}$ and $\{b_1,\ldots,b_t\}$. See Figure~\ref{fig:subsetsum} for an
  illustration.
  
  We claim that there exist $U \in \mathcal{U}_{2,t+2,2}$ and $A \in \mathcal{A}_{t+2,2}$ such that
  $F_0 = U A_0$ and $F_1 = U A_1$ if and only if there exists a subset $B' \subseteq B$
  whose sum is $d$. Equivalently, by Theorem~\ref{thm:combinatorial_characterization}, we claim that $G_\mathcal{F}$ admits a threaded spanning tree $T$
  satisfying \eqref{eq:MSSC} if and only if $B$ has a subset $B'$ whose sum is
  $d$.
  
  We start by proving the forward direction. Since $e = \sum_{\ell=1}^t b_\ell$
  and $T$ is a threaded spanning tree, the sum of the children $\delta(d)$ equals
  $d$ and the sum of the children $\delta(e-d)$ equals $e-d$. Therefore $B' = \delta(d)$ is a subset of $B$ whose sum is $d$. As for
  the reverse direction, we have that the sum of $B'$ equals $d$ and
  thus that the sum of $B \setminus B'$ equals $e-d$. The corresponding 
  threaded spanning tree $T$ where $\delta(*,0) = \{d,e-d\}$, $\delta(d) = B'$ and
  $\delta(e-d) = B \setminus B'$ is therefore a threaded spanning tree of
  $G_\mathcal{F}$ that satisfies \eqref{eq:MSSC}. 
\end{proof}

The result above settles the outstanding question of fixed-parameter
tractability in $m$.

\begin{corollary}
  \Problem~is not fixed-parameter tractable in $m$.
\end{corollary}

\begin{figure}[t]
  \center
  \includegraphics[width=\textwidth]{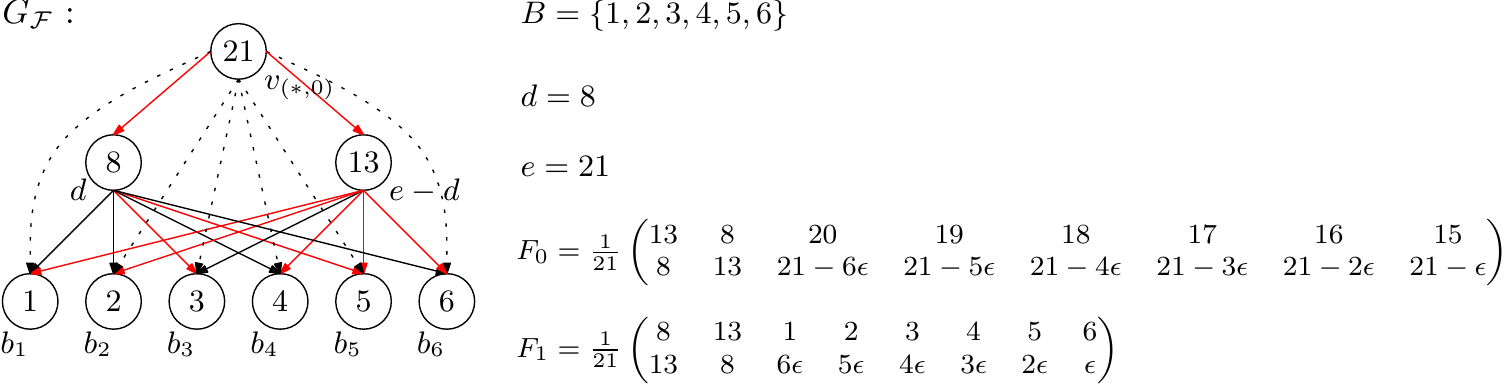}
  \caption{\textbf{Reduction from SUBSET SUM}}
  \label{fig:subsetsum}
\end{figure}

\subsection{\ProblemCladisticFull}

In the case of cladistic characters we are given a set $\mathcal{S} = \{ S_c \mid c \in [n]\}$ of state trees for each character. 
The vertex set of a state tree $S_c$ is $\{0,\ldots,k-1\}$, and the edges describe the relationships of the states of character $c$.
A perfect phylogeny $T$ is \emph{consistent} with $\mathcal{S}$ provided provided $i \prec_{S_{c}} j$ if and only if $(c,i) \prec_T (c,j)$ for all characters $c$ and states $i,j$.

In the cladistic \Problem\ we are given frequency tensor $\mathcal{F}$ and a set of state trees $\mathcal{S} = \{ S_c \mid c \in [n]\}$.
The goal is to find complete perfect phylogeny tree $T$ that generates $\mathcal{F}$ and is consistent with $\mathcal{S}$.
The set of states trees $\mathcal{S}$ determines the set of descendant sets as $D_{(c,i)} = \{j \mid i \prec_{S_c} j\}$.
Therefore the cladistic multi-state ancestry graph $G_{(\mathcal{F}, \mathcal{S})}$ is a simple graph with edges $(v_{(c,i)},v_{(d,j)})$ labeled by $(D_{(c,i)}, D_{(d,j)})$ provided $c \neq d$ and \eqref{eq:MSAC} holds.
Moreover, for each character $c$ there is an edge $(v_{(c,i)}, v_{(c,j)})$ provided state $i$ is the parent of state $j$ in $S_c$.
Solutions of the cladistic \Problem\ correspond to threaded spanning trees in $G_{(\mathcal{F},\mathcal{S})}$ as shown by the following proposition.

\begin{prop1}
  A complete perfect phylogeny tree $T$ generates $\mathcal{F}$ and is consistent with $\mathcal{S}$ if and only if $T$ is a threaded spanning tree of the cladistic ancestry graph $G_{(\mathcal{F},\mathcal{S})}$ such that \eqref{eq:MSSC} holds.
\end{prop1}
\begin{proof}
  Let $T$ be a threaded spanning tree of the ancestry graph $G_{(\mathcal{F},\mathcal{S})}$.
  We claim that $T$ is consistent with $\mathcal{S}$, i.e.\ $i \prec_{S_{c}} j$ if and only if $(c,i) \prec_T (c,j)$ for all characters $c$ and states $i,j$.
  By condition (1) of Definition~\ref{def:threaded}, we have that for all character-state pairs $(c,i) \neq (*,0)$ and $(d,j) \in \delta(c,i)$, adjacent edges $(v_{\pi(c,i)}, v_{(c,i)})$, $(v_{(c,i)}, v_{(d,j)})$ are labeled by $(D_{\pi(c,i)},D_{(c,i)})$ and $(D_{(c,i)},D_{(d,j)})$, respectively---where $\pi(c,i)$ is the parent character-state pair of $(c,i)$.
  By definition, we have $D_{(c,i)} = \{ l \mid i \prec_{S_{c}} l \}$ for each character-state pair $(c,i)$.
  By condition~(2) of Definition~\ref{def:threaded} and the fact that $T$ is a spanning tree of  $G_{(\mathcal{F},\mathcal{S})}$, we have $D_{(c,i)} = \{ l \mid (c,i) \prec_{T} (c,l) \}$.
  The lemma now follows.
  \qed
\end{proof}


Figure~\ref{fig:cladistic_example} shows an example instance $(\mathcal{F},\mathcal{S})$, where $k=3$, $m=2$ and $n=2$, including the cladistic multi-state ancestry graph $G_{(\mathcal{F},\mathcal{S})}$ and a solution $T$ that generates $\mathcal{F}$ and is consistent with $\mathcal{S}$.

\begin{figure}[t]
  \centering
  \subfloat[Frequency tensor {$\mathcal{F} = [F_0 \, F_1 \, F_2]$ where $m=2$, $n=2$ and $k=3$}]{\includegraphics{tensor}}
  \hspace{.5cm}
  \subfloat[State trees {$\mathcal{S}$}]{\includegraphics[scale=1.2]{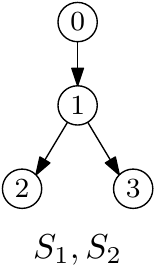}}
  \hspace{.5cm}
  \subfloat[Cladistic multi-state ancestry graph $G_{(\mathcal{F},\mathcal{S})}$, red edges denote $n,k$-complete perfect phylogeny tree $T$ that generates $\mathcal{F}$ and is consistent with $\mathcal{S}$]{\includegraphics[width=.5\textwidth]{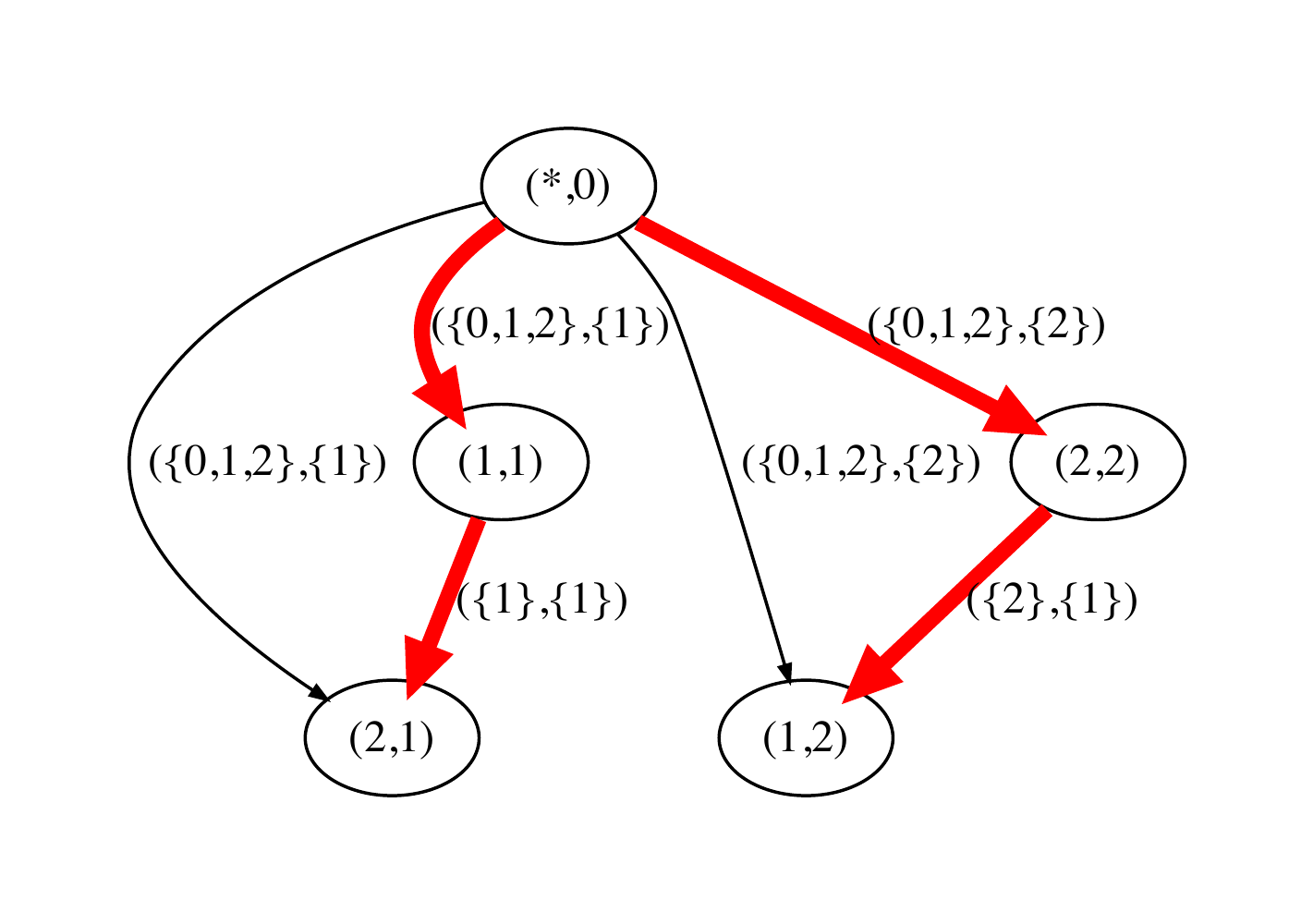}}
  \caption{\textbf{Example cladistic \Problem~instance and solution}}
  \label{fig:cladistic_example}
\end{figure}


\subsection{Algorithm for the \ProblemCladistic}
\label{sec:alg}

We now describe an algorithm for enumerating all trees $T$ that are consistent with the given state trees $\mathcal{S}$ and generate the given frequencies $\mathcal{F}$.
The crucial observation is that any subtree of a consistent, threaded spanning tree $T$ that satisfies \eqref{eq:MSSC} must itself be satisfy \eqref{eq:MSSC} and be consistent and threaded.
A subtree $T'$ is \emph{consistent} if it is rooted at $v_{(*,0)}$, and for each character $c$ the set of states $\{i \mid v_{(c,i)} \in V(T') \}$ induces a connected subtree in $S_c$.
We can thus constructively grow consistent, threaded trees that satisfy \eqref{eq:MSSC} by maintaining the following invariant.

\begin{inv}
  Let tree $T$ be the partially constructed tree.
  It holds that (1) for each $v_{(c,i)} \in V(T) \setminus \{v_{(*,0)}\}$ and parent $\pi(i)$ of $i$ in $S_c$, the vertex $v_{(c,\pi(i))}$ is the first vertex labeled by character $c$ on the unique path from $v_{(*,0)}$ to $v_{(c,i)}$; and (2) for each vertex $v_{(c,i)} \in V(T)$, \eqref{eq:MSSC} holds for $T$ and $\mathcal{F}$.
\end{inv}

We maintain a subset of edges $H \subseteq E(G_{(\mathcal{F},\mathcal{S})})$ called the \emph{frontier} that can be used to extend a partial tree $T$ such that the following invariant holds.
\begin{inv}
  Let tree $T$ be the partially constructed tree.
  For every edge $(v_{(c,i)},v_{(d,j)}) \in H$, (1) $v_{(c,i)} \in V(T)$, (2) $v_{(d,j)} \not \in V(T)$, and (3) Invariant~\ref{inv:T} holds for $T'$ where $E(T') = E(T) \cup \{(v_{(c,i)},v_{(d,j)})\}$.
\end{inv}

Algorithm~\ref{alg:cladistic} gives the pseudo code of \textsc{Enumerate} described in the main text.
The initial call is \textsc{Enumerate}$(G, \{v_{(*,0)}\}, \delta(*,0))$. 
The partial tree containing just the vertex $v_{(*,0)}$ satisfies Invariant~\ref{inv:T}.
The set $\delta(*,0)$ corresponds to the set of outgoing edges from vertex $v_{(*,0)}$ of $G_{(\mathcal{F},\mathcal{S})}$, which by definition satisfies Invariant~\ref{inv:H}.
Upon the addition of an edge $(v_{(c,i)}, v_{(d,j)}) \in H$ (line 5), Invariant~\ref{inv:H} is restored by adding all outgoing edges from $v_{(d,j)}$ whose addition results in a consistent partial tree that satisfies \eqref{eq:MSSC} (lines 8-9) and by removing all edges from $H$ that introduce a cycle (lines 11-12) or violate \eqref{eq:MSSC} (lines 13-14).

\begin{algorithm}
  \caption{\textsc{Enumerate}$(G, T, H)$}
  \label{alg:cladistic}
  \KwIn{Ancestry graph $G_{(\mathcal{F},\mathcal{S})}$, perfect phylogeny tree $T$, frontier $H$}
  \KwOut{All complete perfect phylogeny trees that generate $\mathcal{F}$ and are consistent with $\mathcal{S}$}
  \If{$H = \emptyset~\mathbf{and}~|V(T)| = |V(G)|$}
  {
    Return $T$
  }
  \Else
  {
    \While{$H \neq \emptyset$}
    {
      $(v_{(c,i)}, v_{(d,j)}) \leftarrow \textsc{Pop}(H)$\\
      $E(T) \leftarrow E(T) \cup \{(v_{(c,i)}, v_{(d,j)})\}$\\
      \ForEach{$(v_{(d,j)}, v_{(e,l)}) \in E(G)$}
      {
        \If{$v_{(e,l)} \not \in V(T)~\mathbf{and}~v_{(e,\pi(l))}$ is the first vertex with character $e$ on the path from $v_{(*,0)}$ to $v_{(d,j)}$ $\mathbf{and}~f^+_p(D_{(d,j)}) \geq f^+_p(D_{(e,l)}) + \sum_{(f,s) \in \delta(d,j)} f^+_p(D_{(f,s)})$}
	{
	  \textsc{Push}$(H,(v_{(d,j)}, v_{(e,l)}))$
	}
      }
      \ForEach{$(v_{(e,l)},v_{(f,s)}) \in H$}
      {
	 \If{$v_{(f,s)} = v_{(d,j)}$}
	 {
	   Remove $(v_{(e,l)},v_{(f,s)})$ from $H$\\
	 }
	 \ElseIf{$v_{(e,l)} = v_{(c,i)}~\mathbf{and}~\exists p \in [m] \text{ such that } f^+_p(D_{(c,i)}) < f^+_p(D_{(f,s)}) + \sum_{(h,t) \in \delta(c,i)} f^+_p(D_{(h,t)})$}
	 {
	   Remove $(v_{(e,l)},v_{(f,s)})$ from $H$\\
	 }
      }
      \textsc{Enumerate}$(G, T, H)$\\
      $E(T) \leftarrow E(T) \setminus \{(v_{(c,i)}, v_{(d,j)})\}$
    }
}
\end{algorithm}

We account for errors by taking as input nonempty intervals $[\bunderline{f}_{p,(c,i)}, \boverline{f}_{p,(c,i)}]$ that contain the true frequency $f_{p,(c,i)}$ for character-state pairs $(c,i)$ in samples $p$.
A tree $T$ is \emph{valid} if there exists a frequency tensor $\mathcal{F}' = [[f'_{p,(c,i)}]]$ such that $\bunderline{f}_{p,(c,i)} \le f'_{p,(c,i)} \le \boverline{f}_{p,(c,i)}$ and $\mathcal{F}'$ generates $T$---i.e.\ \eqref{eq:MSSC} holds for $T$ and $\mathcal{F}'$.
A valid tree $T$ is \emph{maximal} if there exists no valid supertree $T'$ of $T$, i.e.\ $E(T) \subsetneq E(T')$.
The task now becomes to find the set of all maximal valid trees.
We recursively define $\hat{\mathcal{F}} = [[\hat{f}_{p,(c,i)}]]$ where
\begin{equation}
  \hat{f}_{p,(c,i)} = \max\Big\{\bunderline{f}_{p,(c,i)}, \sum_{(d,j) \in \delta(c,i)} \hat{f}^+_p(D_{(d,j)}) - \sum_{j \in D(c,i) \setminus \{i\}} \hat{f}_{p,(c,j)}\Big\}.
\end{equation}
The intuition here is to satisfy \eqref{eq:MSSC} by assigning the smallest possible values to the children. We do this bottom-up from the leaves and set $\hat{f}_{p,(c,i)} = \bunderline{f}_{p,(c,i)}$ for each leaf vertex $v_{(c,i)}$.
It turns out that $\bunderline{f}_{p,(c,i)} \leq \hat{f}_{p,(c,i)} \leq \overline{f}_{p,(c,i)}$ is a necessary condition for $T$ to be valid as shown in the  following lemma.

\begin{lem2}
  If tree $T$ is valid then $\bunderline{f}_{p,(c,i)} \le \hat{f}_{p,(c,i)} \le \boverline{f}_{p,(c,i)}$ for all $p$ and $(c,i)$.
\end{lem2}
\begin{proof}
  Let $T$ be a valid tree and let $\mathcal{F} = [[f_{p,(c,i)}]]$ be a frequency tensor that generates $T$ such that $\bunderline{f}_{p,(c,i)} \le f_{p,(c,i)} \le \boverline{f}_{p,(c,i)}$.
  We claim that $\bunderline{f}_{p,(c,i)} \le \hat{f}_{p,(c,i)} \le f_{p,(c,i)} \le \boverline{f}_{p,(c,i)}$.
  We proof this by structural induction on $T$ by working our way up to the root.
  
  For the base case, let $v_{(c,i)}$ be a leaf of $T$. 
  That is, $\delta(c,i) = \emptyset$.
  Thus by definition, $\hat{f}_{p,(c,i)} = \bunderline{f}_{p,(c,i)}$. 
  Therefore, $\bunderline{f}_{p,(c,i)} \le \hat{f}_{p,(c,i)} \le f_{p,(c,i)} \le \boverline{f}_{p,(c,i)}$.
  For the step, let $v_{(c,i)}$ be an inner vertex of $T$.
  Thus, $\delta(c,i) \neq \emptyset$.
  The induction hypothesis (IH) states that $\bunderline{f}_{p,(d,j)} \le \hat{f}_{p,(d,j)} \le f_{p,(d,i)} \le \boverline{f}_{p,(d,j)}$ for all descendants $(d,j)$ of $(c,i)$ in $T$.
  We distinguish two cases:
  \begin{enumerate}
    \item In case $\hat{f}_{p,(c,i)} = \bunderline{f}_{p,(c,i)}$, we have $\bunderline{f}_{p,(c,i)} \le \hat{f}_{p,(c,i)} \le f_{p,(c,i)} \le \boverline{f}_{p,(c,i)}$.
    \item In case $\hat{f}_{p,(c,i)} = \sum_{(d,j) \in \delta(c,i)} \hat{f}^+_p(D_{(d,j)}) - \sum_{j \in D(c,i) \setminus \{i\}} \hat{f}_{p,(c,i)}$, we have $\hat{f}_{p,(c,i)} \ge \bunderline{f}_{p,(c,i)}$ by definition.
      By $\eqref{eq:MSSC}$, we have that $f_{p,(c,i)} \ge \sum_{(d,j) \in \delta(c,i)} f^+_p(D_{(d,j)}) - \sum_{j \in D(c,i) \setminus \{i\}}$.
      By the IH, we have $\bunderline{f}_{p,(d,j)} \le \hat{f}_{p,(d,j)} \le f_{p,(d,j)} \le \boverline{f}_{p,(d,j)}$ for all $(d,j) \in \delta(c,i)$ and $\bunderline{f}_{p,(c,l)} \le \hat{f}_{p,(c,l)} \le f_{p,(c,l)} \le \boverline{f}_{p,(c,l)}$ for all $l \in D_{(c,i)} \setminus \{i\}$.
      Therefore, 
      $\hat{f}_{p,(c,i)} = \sum_{(d,j) \in \delta(c,i)} \hat{f}^+_p(D_{(d,j)}) - \sum_{j \in D(c,i) \setminus \{i\}} \hat{f}_{p,(c,i)} \le \sum_{(d,j) \in \delta(c,i)} f^+_p(D_{(d,j)}) - \sum_{j \in D(c,i) \setminus \{i\}} f_{p,(c,i)} \le f_{p,(c,i)}$.
      Hence, $\bunderline{f}_{p,(c,i)} \le \hat{f}_{p,(c,i)} \le f_{p,(c,i)} \le \boverline{f}_{p,(c,i)}$.
  \qed
  \end{enumerate}
\end{proof}

It may be the case that the frequencies $\hat{f}_{p,(c,i)}$ of a character $c$ in a sample $p$ do not sum to 1. 
Therefore, $\bunderline{f}_{p,(c,i)} \le \hat{f}_{p,(c,i)} \le \boverline{f}_{p,(c,i)}$ is not a sufficient condition for $T$ to be valid. 
However, if $1 - \sum_{i=1}^{k-1} \hat{f}_{p,(c,i)} \le \boverline{f}_{p,(c,0)}$ holds, $T$ is valid as it is generated by frequency tensor $\mathcal{F}' = [[f'_{p,(c,i)}]]$ where 
\begin{equation*}
f'_{p,(c,i)} = 
  \begin{cases}
    1 - \sum_{i=1}^{k-1} \hat{f}_{p,(c,i)}, & \mbox{if $i=0$,}\\
    \hat{f}_{p,(c,i)}, & \mbox{otherwise.}
  \end{cases}
\end{equation*}
This is the case if $\boverline{f}_{p,(c,0)} = 1$.
By assuming the latter, we are able to enumerate all maximal valid trees by updating condition (2) of Invariant~\ref{inv:T} so that \eqref{eq:MSSC} holds for $T$ and $\hat{\mathcal{F}}$.

Algorithm~\ref{alg:noisy_enumerate} gives the pseudo code of an enumeration procedure of all maximal valid trees given intervals $[l_{p,(c,i)}, u_{p,(c,i)}]$ for each character-state pair $(c,i)$ in each sample $p$.
The initial call is \textsc{NoisyEnumerate}$(G, \{v_{(*,0)}\}, \delta(*,0))$.
The partial tree containing just the vertex $v_{(*,0)}$ satisfies Invariant~\ref{inv:T}.
The set $\delta(*,0)$ corresponds to the set of outgoing edges from vertex $v_{(*,0)}$ of $G_{(\mathcal{F},\mathcal{S})}$, which by definition satisfies Invariant~\ref{inv:H}.
Upon the addition of an edge $(v_{(c,i)}, v_{(d,j)}) \in H$ (line 5), Invariant~\ref{inv:H} is restored by adding all outgoing edges from $v_{(d,j)}$ whose addition results in a consistent partial tree $T'$ that satisfies \eqref{eq:MSSC} for $\hat{\mathcal{F}}$ (lines 9-10) and by removing all edges from $H$ that introduce a cycle (lines 12-13) or violate \eqref{eq:MSSC} for $\hat{\mathcal{F}}$ (lines 14-15).
Note that in line 13 we dropped the condition $v_{(e,l)} = v_{(c,i)}$ as the newly added edge $(v_{(c,i)}, v_{(d,j)})$ may affect the frequencies $\hat{\mathcal{F}}$ of the vertices of the current partial tree $T$.

Since a maximal valid tree $T$ does not necessarily span all the vertices, it may happen that for a character $c$ not all states in $S_c$ are present.
We say that a maximal valid tree $T$ is \emph{state complete} if for each vertex $v_{(c,i)}$ of $T$, all vertices $v_{(c,j)}$ where $j \in V(S_c)$ are also in $V(T)$.
Our goal is to report all maximal valid and state-complete trees.
Therefore, we post-process each maximal valid tree $T$ and remove all vertices $v_{(c,i)}$ where there is a $j \in V(S_c)$ such that $v_{(c,j)} \not \in V(T)$.
The tree that we report corresponds to the connected component rooted at $v_{(*,0)}$.
Since each maximal valid and state-complete tree is a partial valid tree rooted at $v_{(*,0)}$, our enumeration procedure reports all maximal valid and state-complete trees.

\begin{algorithm}
  \caption{\textsc{NoisyEnumerate}$(G, T, H)$}
  \label{alg:noisy_enumerate}
  \KwIn{Ancestry graph $G_{(\mathcal{F},\mathcal{S})}$, perfect phylogeny tree $T$, frontier $H$}
  \KwOut{All maximal valid perfect phylogenies that are consistent with $\mathcal{S}$}
  \If{$H = \emptyset$}
  {
    Let $T'$ be the subtree of $T$ that only contains state-complete characters\\
    Return $T'$
  }
  \Else
  {
    \While{$H \neq \emptyset$}
    {
      $(v_{(c,i)}, v_{(d,j)}) \leftarrow \textsc{Pop}(H)$\\
      $E(T) \leftarrow E(T) \cup \{(v_{(c,i)}, v_{(d,j)})\}$\\
      \ForEach{$(v_{(d,j)}, v_{(e,l)}) \in E(G)$}
      {
	\If{$v_{(e,l)} \not \in V(T)~\mathbf{and}~v_{(e,\pi(l))}$ is the first vertex with character $e$ on the path from $v_{(*,0)}$ to $v_{(d,j)}$ $\mathbf{and}~\hat{f}^+_p(D_{(d,j)}) \geq \hat{f}^+_p(D_{(e,l)}) + \sum_{(f,s) \in \delta(d,j)} \hat{f}^+_p(D_{(f,s)})$}
	{
	  \textsc{Push}$(H,(v_{(d,j)}, v_{(e,l)}))$
	}
      }
      \ForEach{$(v_{(e,l)},v_{(f,s)}) \in H$}
      {
	 \If{$v_{(f,s)} = v_{(d,j)}$}
	 {
	   Remove $(v_{(e,l)},v_{(f,s)})$ from $H$\\
	 }
	 \ElseIf{$\exists p \in [m], \hat{f}^+_p(D_{(c,i)}) < \hat{f}^+_p(D_{(f,s)}) + \sum_{(h,t) \in \delta(c,i)} \hat{f}^+_p(D_{(h,t)})$}
	 {
	   Remove $(v_{(e,l)},v_{(f,s)})$ from $H$\\
	 }
      }
      \textsc{NoisyEnumerate}$(G, T, H)$\\
      $E(T) \leftarrow E(T) \setminus \{(v_{(c,i)}, v_{(d,j)})\}$
    }
  }
\end{algorithm}

\subsection{Application to Cancer Sequencing Data}
\label{sec:realdata}

Figure~\ref{fig:LOH_SCD_SNV} shows the eleven state trees that satisfy the assumption where at each locus at most one CNA event occurs (either an LOH or an SCD) as well as at most one SNV event.  
Table~\ref{tab:CNA} shows the relative frequency assignments for each state, described in Section~\ref{sec:cancer_model}.
Table~\ref{tbl:intervals} shows the allowed VAFs $[\bunderline{h}, \boverline{h}]$ given a state tree $S$ and mixing proportions $\mu$.

\begin{figure}
  \center
  \includegraphics[width=.9\textwidth]{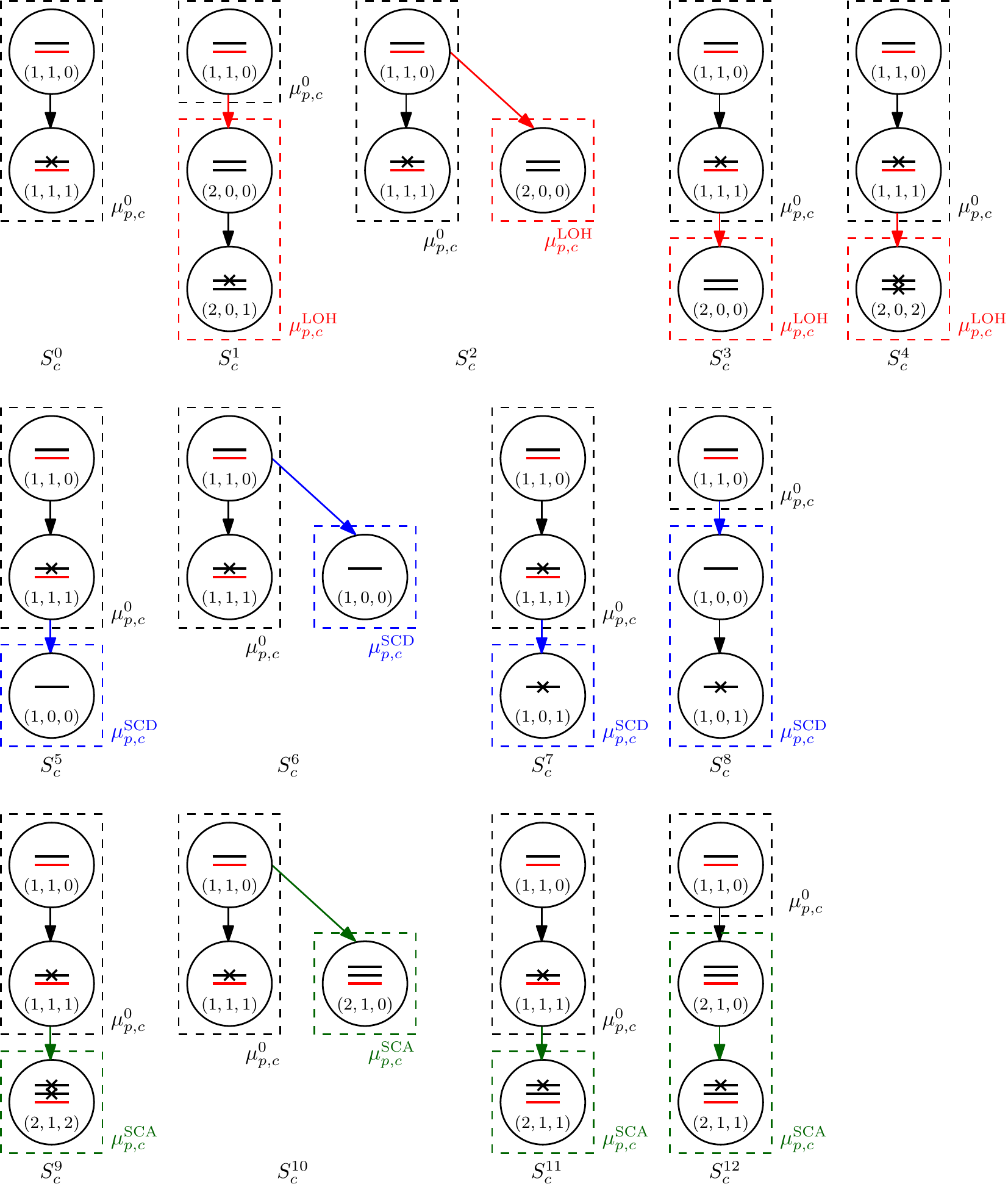}
  \caption{\textbf{Allowed state trees.} Red edges indicate LOH events, blue edges indicate SCD events, green edges indicate SCA events, and black edges indicate SNV events.  State tree $S_c^0$ corresponds to a single SNV event. State trees $S_c^1, S_c^2,S_c^3,S_c^4$ correspond to mixed LOH/SNV events, state trees $S_c^5,S_c^6,S_c^7,S_c^8$ to mixed SCD/SNV events and states trees $S_c^9,S_c^{10},S_c^{11},S_c^{12}$ to mixed SCA/SNV events.}
  \label{fig:LOH_SCD_SNV}
\end{figure}

\begin{sidewaystable}[t]
  \center
  \caption{\textbf{Relative frequencies.} Each state tree $S_c \in \Sigma$, defines a linear system of equations with one solution in terms of variables $f_{p,(c,i)}$ and constants $h_{p,c}$, $\mu^0_{p,c}$, $\mu^\mathrm{LOH}_{p,c}$, $\mu^\mathrm{SCD}_{p,c}$ and $\mu^\mathrm{SCA}_{p,c}$.}
  \label{tab:CNA}
  \begin{tabular}{|l|l|l|l|l|l|l|l|l|l|l|}
    \hline
    $\Sigma$ & $f_0 = f_{(1,1,0)}$ & $f_1 = f_{(1,1,1)}$ & $f_2 = f_{(2,0,0)}$ & $f_3 = f_{(2,0,1)}$ & $f_4 = f_{(2,0,2)}$ & $f_5 = f_{(1,0,0)}$ & $f_6 = f_{(1,0,1)}$ & $f_7 = f_{(2,1,0)}$ & $f_8 = f_{(2,1,1)}$ & $f_9 = f_{(2,1,2)}$\\
    \hline
    \hline
    $S^0$ & $1- h$ & $h$ & 0 & 0 & 0 & 0 & 0 & 0 & 0 & 0\\
    \hline
    $S^1$ & $\mu^0$ & 0 & $\mu^\mathrm{LOH} - 2 h$ & $2 h$ & 0 & 0 & 0 & 0 & 0 & 0\\
    $S^2$ & $\mu^0 - 2 h$ & $2 h$ & $\mu^\mathrm{LOH}$ & 0 & 0 & 0 & 0 & 0 & 0 & 0\\
    $S^3$ & $\mu^0 - 2 h$ & $2 h$ & $\mu^\mathrm{LOH}$ & 0 & 0 & 0 & 0 & 0 & 0 & 0\\
    $S^4$ & $\mu^0 - (2 h - 2 \mu^\mathrm{LOH})$ & $2 h - 2 \mu^\mathrm{LOH}$ & 0 & 0 & $\mu^\mathrm{LOH}$ & 0 & 0 & 0 & 0 & 0\\
    \hline
    $S^5$ & $\mu^0 - \frac{\mu^\mathrm{SCD}}{ h}$ & 0 & 0 & 0 & $\frac{\mu^\mathrm{SCD}}{ h}$&  $\mu^\mathrm{SCD}$ & 0 & 0 & 0 & 0\\
    $S^6$ & $\mu^0 - \frac{\mu^\mathrm{SCD}}{ h}$ & 0 & 0 & 0 & $\frac{\mu^\mathrm{SCD}}{ h}$ & $\mu^\mathrm{SCD}$ & 0 & 0 & 0 & 0\\
    $S^7$ & $\mu^0 - h(1+\mu^0) - \mu^\mathrm{SCD}$ & $h(1+\mu^0) - \mu^\mathrm{SCD}$ & 0 & 0 & 0 & 0 & $\mu^\mathrm{SCD}$ & 0 & 0 & 0\\
    $S^8$ & $\mu^0$ & 0 & 0 & 0 & 0 & $\mu^\mathrm{SCD} - h(1+\mu^0) $  & $h(1+\mu^0)$ & 0 & 0 & 0\\
    \hline
    \hline
    $S^9$ & $1 + \mu^\mathrm{SCA} - h(2 + \mu^\mathrm{SCA})$ & $h(2 + \mu^\mathrm{SCA}) - 2 \mu^\mathrm{SCA}$ & 0 & 0 & 0 & 0 & 0 & 0 & 0 & $\mu^\mathrm{SCA}$ \\
    $S^{10}$ & $\mu^0 - h(2 + \mu^\mathrm{SCA})$ & $h(2 + \mu^\mathrm{SCA})$ & 0 & 0 & 0 & 0 & 0 & $\mu^\mathrm{SCA}$ & 0 & 0\\
    $S^{11}$ & $1 - h(2 + \mu^\mathrm{SCA})$ & $h(2 + \mu^\mathrm{SCA}) - \mu^\mathrm{SCA}$ & 0 & 0 & 0 & 0 & 0 & 0 & $\mu^\mathrm{SCA}$ & 0\\
    $S^{12}$ & $\mu^0$ & 0 & 0 & 0 & 0 & 0 & 0 & $\mu^\mathrm{SCA} - h(2 + \mu^\mathrm{SCA})$  & $h(2 + \mu^\mathrm{SCA})$ & 0\\
    \hline
  \end{tabular}
\end{sidewaystable}

\begin{table}[t]
  \center
  \caption{\textbf{Allowed VAFs.} Given CNA proportions $\mu^0_{p,c}$, $\mu^\mathrm{LOH}_{p,c}$, $\mu^\mathrm{SCD}_{p,c}$ and $\mu^\mathrm{SCA}_{p,c}$, each state tree $S_c \in \Sigma$ determines the interval $[\bunderline{h},\overline{h}]$ of allowed VAFs.}
  \label{tbl:intervals}
  \begin{tabular}{|l|l|}
    \hline
    $\Sigma$ & $[\bunderline{h}, \overline{h}]$\\
    \hline
    \hline
    $S^0$ & $[0, 0.5]$\\
    \hline
    $S^1$ & $\left[0,\frac{\mu^\mathrm{LOH}}{2}\right]$\\
    $S^2$ & $\left[0,\frac{\mu^0}{2}\right]$\\
    $S^3$ & $\left[0,\frac{\mu^0}{2}\right]$\\
    $S^4$ & $\left[\mu^\mathrm{LOH}, \frac{1 + \mu^\mathrm{LOH}}{2}\right]$\\
    \hline
    $S^5$ & $\left[0, \frac{\mu^0}{1 + \mu^0} \right]$\\
    $S^6$ & $\left[0, \frac{\mu^0}{1 + \mu^0} \right]$\\
    $S^7$ & $\left[\frac{\mu^\mathrm{SCD}}{1 + \mu^0}, \frac{1}{1 + \mu^0} \right]$\\
    $S^8$ & $\left[0, \frac{\mu^\mathrm{SCD}}{1 + \mu^0} \right]$\\
    \hline
    $S^9$ & $\left[\frac{2 \mu^\mathrm{SCA}}{2 + \mu^\mathrm{SCA}}, \frac{1 + \mu^\mathrm{SCA}}{2 + \mu^\mathrm{SCA}} \right]$\\
    $S^{10}$ & $\left[0, \frac{\mu^0}{2 + \mu^\mathrm{SCA}}\right]$\\
    $S^{11}$ & $\left[\frac{\mu^\mathrm{SCA}}{2 + \mu^\mathrm{SCA}}, \frac{1}{2 + \mu^\mathrm{SCA}}\right]$\\
    $S^{12}$ & $\left[0, \frac{\mu^\mathrm{SCA}}{2 + \mu^\mathrm{SCA}}\right]$\\
    \hline
  \end{tabular}
\end{table}

\section{Supplementary Results}

We include here additional results on both simulated data and real data not included in the main text. 

\begin{figure}[t]
  \centering
  \subfloat[\label{fig:res_D}Number of solutions (log-scale) for noisy data with $n=4$]{\includegraphics[width=.93\textwidth]{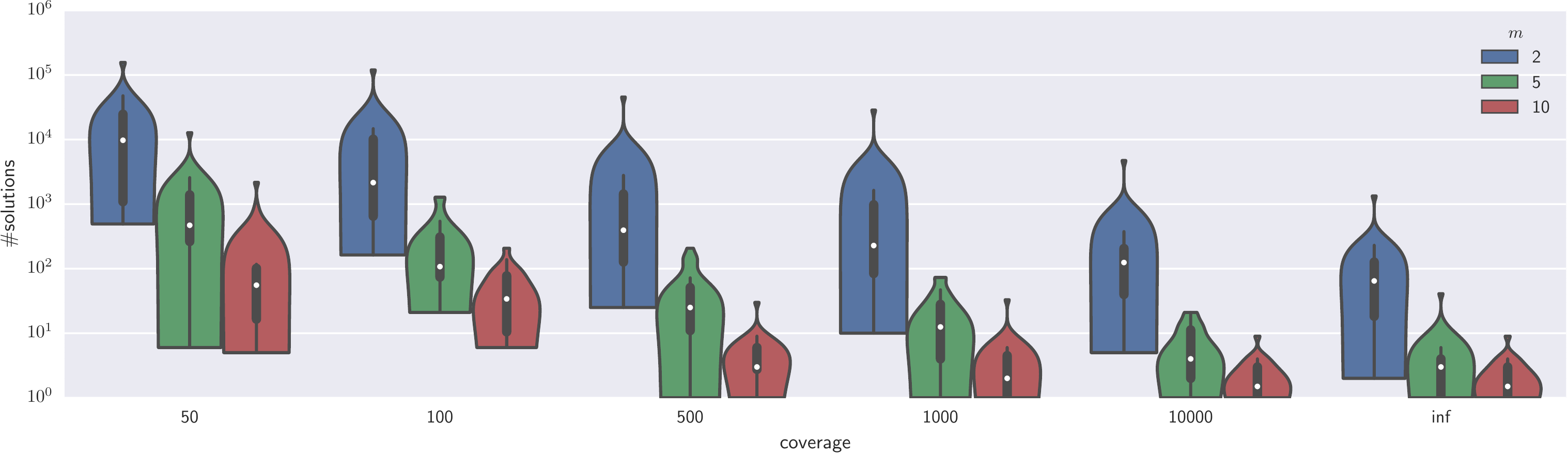}}
  \\
  \subfloat[\label{fig:res_D3}Running time (s, log-scale) for noisy data with $n=4$]{\includegraphics[width=.93\textwidth]{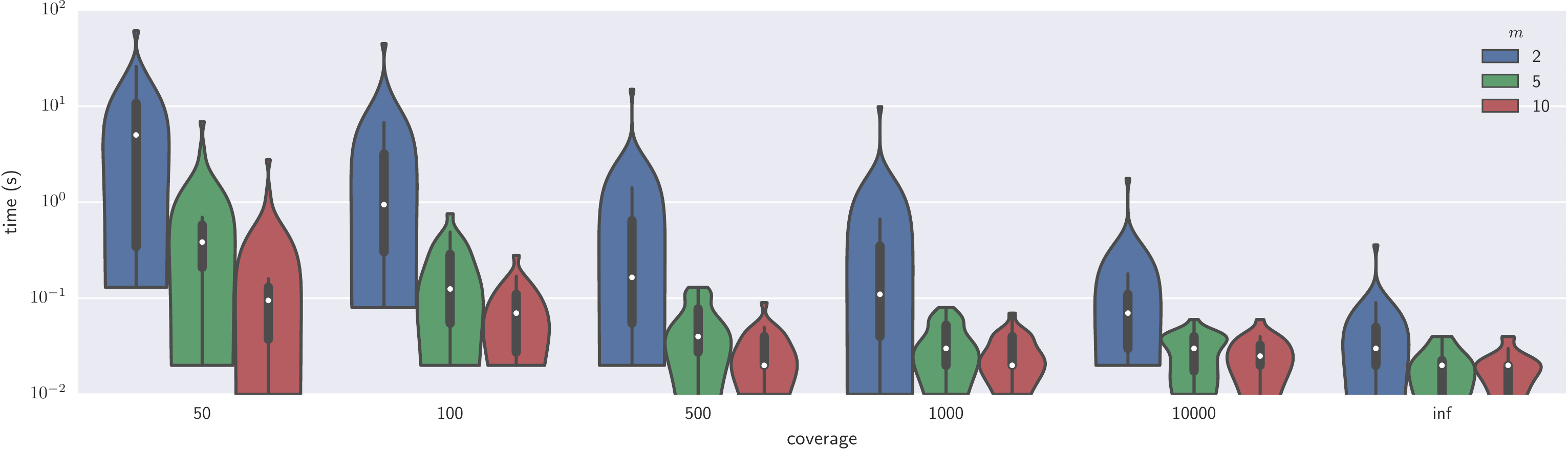}}
  \caption{\textbf{Simulated data results:} running time (seconds, log-scale) and number of solutions (log-scale). A coverage of `inf' corresponds to error-free VAFs (Figure~\ref{fig:res_A})}
  \label{fig:noisy}
\end{figure}

\begin{figure}[t]
  \centering
  \includegraphics[width=.7\textwidth]{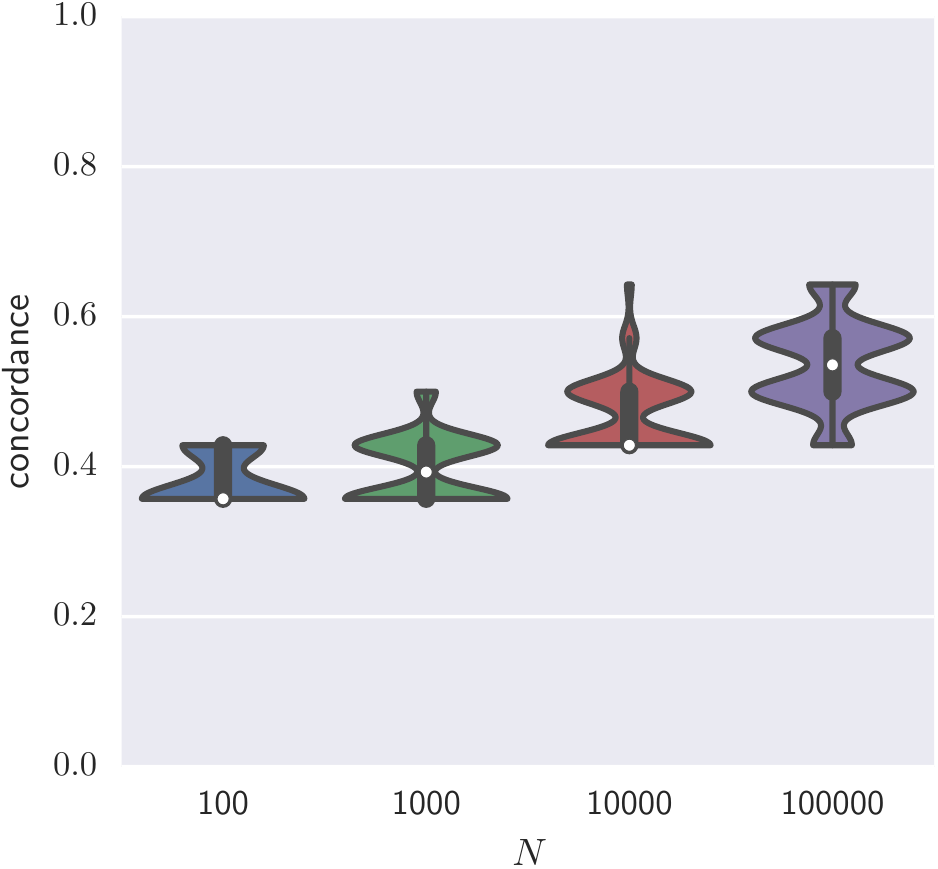}
  \caption{\textbf{Simulated data results.} Concordance with the simulated tree of an instance with $n=10$, $m=10$ and target coverage of 1000x}
  \label{fig:res_E}
\end{figure}

Figure~\ref{fig:solution_space_prostate} shows the solution space of tumor A22 ($N=1000000$).

\begin{figure}[t]
  \center
  \subfloat[Tumor CLL077 (20 solutions).  We note GPR158 is not show here as it was not contained in any of the 30 solutions.]{\includegraphics[width=0.8\textwidth]{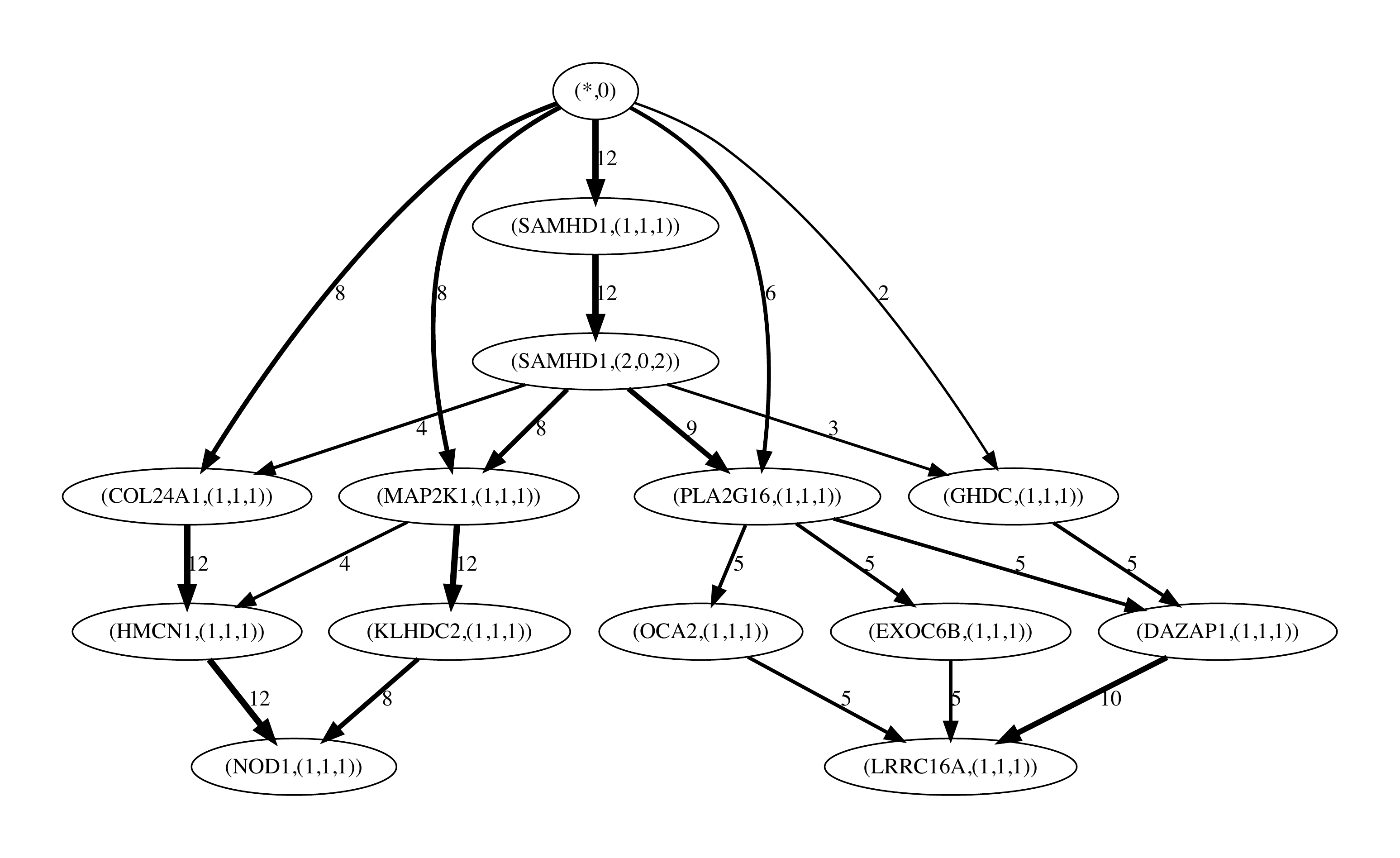} \label{fig:solution_space_cll}}
  \\
  \subfloat[Tumor A22 (24288 trees)]{\includegraphics[width=0.8\textwidth]{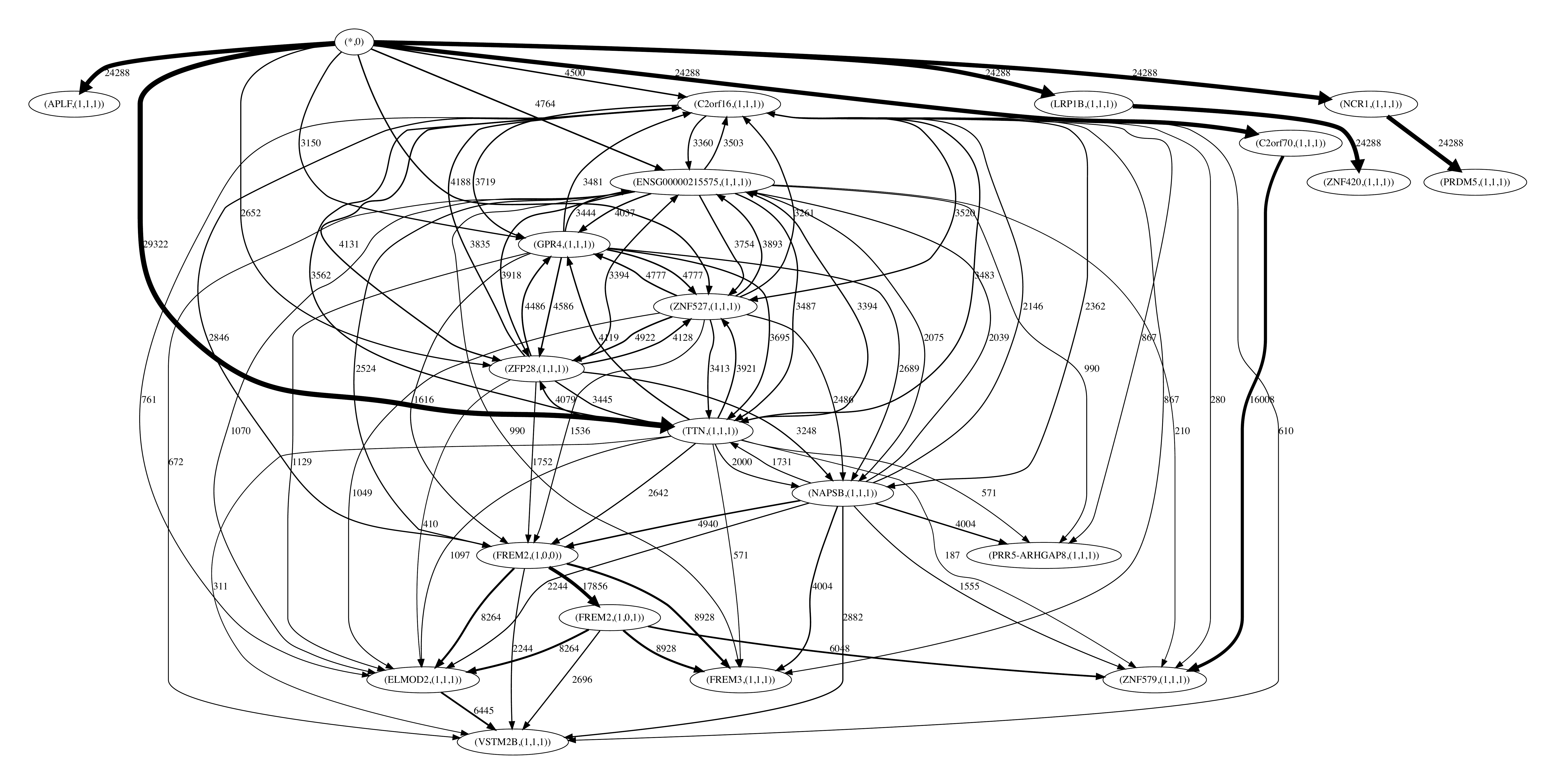}}
  \caption{\textbf{Enumerated solution space of real data instances.} Vertices correspond to the vertices of the solution trees and each edge is labeled by the number of solutions in which it occurs}
  \label{fig:solution_space_prostate}
\end{figure}

\end{document}